\documentclass[nohyperref]{article}
\usepackage{hyperref}       

\usepackage[accepted]{icml2022}

\usepackage[utf8]{inputenc} 
\usepackage[T1]{fontenc}    

\usepackage{url}            
\usepackage{booktabs}       
\usepackage{amsfonts}       
\usepackage{nicefrac}       
\usepackage{microtype}      
\usepackage{xcolor}         
\usepackage{wrapfig}
\usepackage{bm}
\usepackage{amsmath}
\usepackage{amssymb}
\usepackage{amsthm}
\usepackage{amsmath}
\usepackage{graphicx}
\usepackage{subfigure}
\usepackage{enumitem}
\usepackage{diagbox}
\usepackage{multirow}
\usepackage{apptools}
\usepackage{diagbox}
\usepackage{makecell}

\theoremstyle{plain}
\newtheorem{theorem}{Theorem}[section]
\newtheorem{proposition}[theorem]{Proposition}
\newtheorem{lemma}[theorem]{Lemma}
\newtheorem{corollary}[theorem]{Corollary}
\theoremstyle{definition}
\newtheorem{definition}[theorem]{Definition}

\theoremstyle{remark}

\definecolor{azure}{rgb}{0.0, 0.5, 1.0}
\definecolor{brandeisblue}{rgb}{0.0, 0.44, 1.0}
\definecolor{darkpastelgreen}{rgb}{0.01, 0.75, 0.24}
\definecolor{darkpastelpurple}{rgb}{0.59, 0.44, 0.84}
\definecolor{darktangerine}{rgb}{1.0, 0.66, 0.07}
\definecolor{debianred}{rgb}{0.84, 0.04, 0.33}
\definecolor{hanpurple}{rgb}{0.32, 0.09, 0.98}
\definecolor{deepcerise}{rgb}{0.85, 0.2, 0.53}
\definecolor{emerald}{rgb}{0.31, 0.78, 0.47}
\definecolor{fuchsia}{rgb}{1.0, 0.0, 1.0}
\definecolor{flamingopink}{rgb}{0.99, 0.56, 0.67}
\definecolor{lightseagreen}{rgb}{0.13, 0.7, 0.67}
\definecolor{mayablue}{rgb}{0.45, 0.76, 0.98}
\newcommand{\col}{\textcolor{brandeisblue}}

\newcommand{\EE}{\mathbb{E}}

\newcommand{\V}{\mathcal{V}}

\newcommand{\R}{\mathbb{R}}

\DeclareMathOperator*{\argmin}{arg\,min}
\DeclareMathOperator*{\diag}{diag}

\DeclareMathOperator*{\corr}{corr}

\DeclareMathOperator*{\supp}{supp}
\DeclareMathOperator*{\id}{Id}

\icmltitlerunning{Consensus Multiplicative Weights Update: Learning to Learn using Projector-based Game Signatures}

\begin{document}
\twocolumn[
\icmltitle{Consensus Multiplicative Weights Update: Learning to Learn using Projector-based Game Signatures}
\begin{icmlauthorlist}
\icmlauthor{Nelson Vadori}{yyy}
\icmlauthor{Rahul Savani}{zzz}
\icmlauthor{Thomas Spooner}{yyy}
\icmlauthor{Sumitra Ganesh}{yyy}
\end{icmlauthorlist} 

\icmlaffiliation{yyy}{J.P. Morgan AI Research}

\icmlaffiliation{zzz}{Department of Computer Science, University of Liverpool}

\icmlcorrespondingauthor{Nelson Vadori}{nelson.n.vadori@jpmorgan.com}
\icmlkeywords{game theory, multiplicative weights update, consensus optimization, game decomposition, learning to learn}
\vskip 0.3in
]
 
\printAffiliationsAndNotice{}

\begin{abstract}
Cheung and Piliouras (2020) recently showed that two variants of the Multiplicative Weights Update method - OMWU and MWU - display opposite convergence properties depending on whether the game is zero-sum or cooperative. Inspired by this work and the recent literature on learning to optimize for single functions, we introduce a new framework for learning last-iterate convergence to Nash Equilibria in games, where the update rule's coefficients (learning rates) along a trajectory are learnt by a reinforcement learning policy that is conditioned on the nature of the game: \textit{the game signature}. We construct the latter using a new decomposition of two-player games into eight components corresponding to commutative projection operators, generalizing and unifying recent game concepts studied in the literature. We compare the performance of various update rules when their coefficients are learnt, and show that the RL policy is able to exploit the game signature across a wide range of game types. In doing so, we introduce CMWU, a new algorithm that extends consensus optimization to the constrained case, has local convergence guarantees for zero-sum bimatrix games, and show that it enjoys competitive performance on both zero-sum games with constant coefficients and across a spectrum of games when its coefficients are learnt.
\end{abstract}

\section{Introduction}
\label{sec1}

Motivated by applications in game theory, optimization, and Generative Adversarial Networks, the machine learning community has recently started using variants of the gradient method, such as optimistic gradient \cite{oo} and consensus optimization \cite{co}. These methods enjoy linear convergence in cases like (unconstrained) bilinear games, where standard gradient descent fails (see also \cite{gidel1,gidel2}). In the simplex-constrained setting where "strategy weights" are restricted to be non-negative and sum to one, deriving constant step-size \footnote{in the sense of time-independence.} learning algorithms that display last-iterate convergence is difficult: to the best of our knowledge, existing algorithms include Optimistic Multiplicative Weights Update (OMWU), Optimistic Mirror Descent (OMD) and Optimistic Gradient Descent Ascent (OGDA). Such convergence guarantees with constant step-size were obtained for zero-sum bimatrix games in the case of OMWU \cite{omwu}, and then extended to the convex-concave zero-sum setting \cite{omwu2}; for so-called "coherent" saddle point problems with a Lipschitz condition in the case of OMD, which includes the zero-sum bimatrix case \cite{merti} (later generalized in \cite{hsieh}); for zero-sum games satisfying a so-called Saddle-Point Metric Subregularity condition in the case of OGDA, which includes the bimatrix and convex-concave setting \cite{ogda}. The difficulty is in part due to simplex constraints on the weights: tweaking one coordinate impacts all others, making it difficult for such an update rule to admit Nash Equilibria (NE) as fixed points. Last-iterate convergence is sometimes referred to as "pointwise", and has to be contrasted with convergence in the time-averaged sense.

In this work, we introduce \textbf{\col{Consensus Multiplicative Weights Update (CMWU)}}, a new algorithm that extends consensus optimization to the constrained case. It enjoys local last-iterate convergence to NE in the zero-sum bimatrix case, using constant learning rates. In contrast to OMWU, its update rule from $t$ to $t+1$ only depends on players' strategies at time $t$. In contrast to OMD and OGDA, CMWU does not require the computation of mirror coordinates, although it requires a second order term. OMWU, CMWU and OMD all use softmax projection, but OGDA requires an orthogonal projection onto the unit simplex with respect to the $L_2$ norm, which requires vector sorting. 

CMWU, like OMWU, OMD and OGDA, has convergence guarantees, for bimatrix games, in the zero-sum case only. Recent work \cite{chaos} has shown that MWU and OMWU display opposite convergence properties depending on whether the game is zero-sum or cooperative, where players' payoffs are equal: OMWU (resp. MWU) converges (resp. diverges) in the zero-sum case, and vice-versa in the cooperative case. Since OMWU reduces to MWU with zero coefficient on the gradient difference term $\nabla_t-\nabla_{t-1}$, another way to see this opposite behavior is that if one chooses suitably the update rule's coefficients as a function of the nature of the game, one can hope to get convergence across a spectrum of games. The issue, however, is that a general two-player (2P) game is not purely zero-sum or cooperative, hence we need to answer the question: \textbf{\col{\textit{how to encode a 2P game?}}} Recent literature has studied this question in specific settings: \cite{2pszs} shows that every 2P zero-sum symmetric game can be split into the direct sum of a transitive and a cyclic component. \cite{chaos} introduces trivial games where the payoff is decoupled between the 2 players, and a related triviality metric that allows them to compute to which extent a game is trivial. In the game theory literature, it has been observed that commutative projectors can be used to construct (direct sum) game decompositions, for example \cite{proj} formulates decompositions into 2 and 3 components. We build on this work to formulate a new canonical decomposition of a 2P game into $2^n$ components, given $n$ commutative projectors, and apply this result to the case $n=3$, with projectors associated to zero-sum/cooperative, symmetric/antisymmetric and transitive/cyclic games. Importantly, this view generalizes the concept of transitivity in \cite{2pszs} to any 2P game, and in addition we show that trivial games in \cite{chaos} in fact emanate from the same projector as transitive games. This unified view allows us to view trivial games in terms of their mirror, cyclic games, and gives an alternative to the triviality metric in \cite{chaos}. Based on this decomposition, we introduce a new concept, the \textbf{\col{\textit{game signature}}}, which we use as a \textit{feature} to learn the update rule's coefficients as a function of the game type, extending the learning to optimize (L2O) framework \cite{l2o} to games.

\textbf{Contributions}. \textbf{(1)} We introduce CMWU, a new constant step-size algorithm that enjoys local last-iterate convergence to NE in the constrained zero-sum bimatrix case, for which we show empirically that it enjoys competitive performance with the best known methods, see for example in figure \ref{rps} the case of rock-paper-scissors. \textbf{(2)} We introduce a new concept, the game signature, based on a new canonical decomposition of 2-player games into $2^n$ components, given $n$ commutative projectors, and apply this result to the case $n=3$, with projectors associated to zero-sum/cooperative, symmetric/antisymmetric and transitive/cyclic games. Importantly, this view generalizes the decomposition of \cite{2pszs} to any 2P game, and unifies their transitive games with the trivial games of \cite{chaos}. \textbf{(3)} Based on (1) and (2), we introduce a new framework for last-iterate convergence learning in games, where the algorithm update rule's coefficients along a trajectory are learnt by a reinforcement learning policy which is conditioned on the game signature. To our knowledge this is the first application of L2O to finding NE based on the nature of the game. We show empirically that our learned policy is able to exploit the game signature across a wide range of game types, that CMWU performs well in that framework, where we learn to exploit the gradient in transitive games and the Hessian in cyclic ones, and learn mirror behaviors across zero-sum and cooperative components.

\begin{figure}[th]
  \centering
\centerline{\includegraphics[scale=0.5]{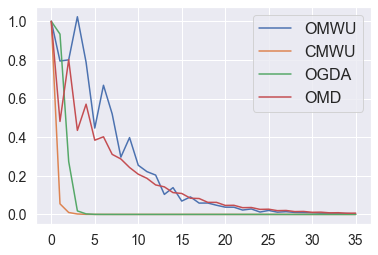}}
  \caption{Rock-paper-scissors. Last-iterate convergence to Nash $\frac{\delta(x^t,y^t)}{ \delta(x^0,y^0)}$. Best hyperparameters.}
  \label{rps}
\end{figure}

\section{CMWU: Consensus Multiplicative Weights Update}
\label{seccmwu}
\textbf{Setting and notations}. We consider a zero-sum bimatrix game where players receive payoffs $x^TAy$ and $-x^TAy$, $A \in \mathcal{M}_{n,m}$ is a matrix of size $n \times m$, $x \in \Delta_n$, $y \in \Delta_m$ the vectors of players' mixed strategies, and $\Delta_k$ the probability simplex in dimension $k$ \footnote{set of vectors with nonnegative coordinates summing to one.}. We define $v_A:=x^{*T}Ay^*$ the value of the game, where $(x^{*},y^{*})$ is a NE \footnote{it is a classical result that at least one NE exists, and that $v_A$ is the same across all possible NE.}. For $z \in \R^k$, we denote $D_z:=\diag(z)$ the diagonal matrix with $z$ on the diagonal.

\textbf{Key insights}. We recall that in the unconstrained case, players' gradients are perturbed by the Hessian terms $AA^Tx$, $A^TAy$ to yield convergence, cf. details in appendix. There are two main insights in our extension of consensus optimization to the constrained setting. \col{\textbf{(1)}} what is the constrained equivalent of the unconstrained Hessian $AA^T$, and of its invertibility requirement guaranteeing its positive-definiteness so as to "bend" the eigenvalues of the game Jacobian below one? \col{\textbf{(2)}} how can we design a Hessian-based update rule that admits any NE as a fixed point? For \col{\textbf{(1)}}, we introduce in definition \ref{hess} "renormalized" versions $H_x:=A^TD_xA$, $H_y:=AD_yA^T$ of the classical Hessian that we call \textbf{\col{simplex-Hessians}}. One way to see this renormalization is from the fact that, by commutativity of the spectrum of a matrix product, $sp(AB)=sp(BA)$, we can see the spectrum of $H_x$ and $H_y$ as scaled versions of that of $AA^T$. Related to invertibility, games like rock-paper-scissors do not have an invertible payoff matrix. Interestingly, our local convergence theorem \ref{t1} requires a \textbf{\col{weak invertibility}} condition (definition \ref{wi}), where weak is both in the sense that it accommodates non-square matrices, and that it requires invertibility only on the subspace of vectors whose coordinates sum to zero. This is a relatively weak condition that explains why the algorithm works well for a large class of payoffs, including rock-paper-scissors. Finally, for \col{\textbf{(2)}}, a key observation is that equilibrium properties specifically make all NE be fixed points of the simplex-Hessian based CMWU update (proposition \ref{sp}). Due to the simplex constraint, it is quite difficult to find update rules that admit NE as fixed points, since tweaking one coordinate impacts the others.

\begin{definition}
\label{hess}
\col{\textbf{(simplex-Hessians)}} The simplex-Hessians $\col{H_x} \in \mathcal{M}_{m,m}$, $\col{H_y} \in \mathcal{M}_{n,n}$ associated to $x \in \Delta_n$, $y \in \Delta_m$ are defined as $H_x:=A^TD_xA$, $H_y:=AD_yA^T$.
\end{definition}

\textbf{The CMWU update.} Simplex-Hessians are symmetric, square matrices by construction. We use them to extend consensus optimization to the constrained case. Let $h>0$ be the learning rate and $\epsilon>0$ the simplex-Hessian coefficient. We consider the multiplicative update $(x^{t+1},y^{t+1})=(\varphi_1(x^{t},y^{t}),\varphi_2(x^{t},y^{t}))$ at iteration stage $t$ in (\ref{cmwu}), for $i \in [1,n]$ and $j \in [1,m]$. Note that, for $\epsilon=0$, we recover the classical MWU. The key observation is proposition \ref{sp} that shows -- using equilibrium properties only -- that any NE is a fixed point of CMWU.

\begin{align}
    \label{cmwu}
    \begin{split}
    &[\varphi_1(x,y)]_i:=  \frac{x_i\exp(h[Ay]_i - h\epsilon [\col{H_y}x]_i)}{\sum_{k=1}^n x_k\exp(h[Ay]_k - h\epsilon [\col{H_y} x]_k)}\\
    &[\varphi_2(x,y)]_j:= \frac{y_j\exp(-h[A^Tx]_j - h\epsilon [\col{H_x}y]_j)}{\sum_{k=1}^m y_k\exp(-h[A^Tx]_k - h\epsilon [\col{H_x} y]_k)}
    \end{split}
\end{align}

\begin{proposition}
\label{sp}
\col{\textbf{(Fixed point property of CMWU)}} Any NE $(x^{*},y^{*})$ of the zero-sum bimatrix game $A$ is a fixed point of the CMWU update (\ref{cmwu}), namely $\varphi_1(x^*,y^*)=x^*$, $\varphi_2(x^*,y^*)=y^*$.
\end{proposition}
\begin{proof}
If $x_i^*=0$, $[\varphi_1(x^*,y^*)]_i=x_i^*=0$. By property of the NE \cite{omwu}, $x_i^*\neq 0 \Rightarrow [Ay^*]_i=v_A$, and $y_j^*\neq 0 \Rightarrow [A^Tx^*]_j=v_A$. Hence if $x_i^* \neq 0$, $[H_{y^*}x^*]_i=[A D_{y^*}A^Tx^*]_{i}=\sum_{j=1}^m A_{ij} y^*_j [A^T x^*]_j=\sum_{j:y^*_j\neq 0} A_{ij} y^*_j [A^T x^*]_j = [A y^*]_i v_A= v_A^2$, and so $[\varphi_1(x^*,y^*)]_i=x^*_i \frac{\exp(hv_A - h\epsilon v_A^2)}{\exp(hv_A - h\epsilon v_A^2)\sum_{k=1}^n x_k}=x_i^*$. Similarly $\varphi_2(x^*,y^*)=y^*$ since for $y_j^* \neq 0$, $[H_{x^*}y^*]_j=v_A^2$.
\end{proof}

\begin{definition}
\label{rg}
\col{\textbf{(quasi-strict NE)}} A NE $(x^*,y^*)$ of the zero-sum game $A$ is said to be quasi-strict if $[Ay^*]_i<v_A$ whenever $x_i^*=0$, and $[A^Tx^*]_j>v_A$ whenever $y_j^*=0$.
\end{definition}

Quasi-strict equilibria were introduced originally in \cite{quasistrict}. Every interior NE \footnote{i.e. one where both players play all of their available actions with positive probability.} is by construction quasi-strict, and every bimatrix game has at least one quasi-strict equilibrium \cite{reg,quasistrict4,quasistrict3}. In general, bimatrix games can have multiple quasi-strict equilibria. The proof of local convergence for OMWU \cite{omwu} requires uniqueness of the NE, however in our case we only need the weaker quasi-strict property to handle the non-interior equilibrium coordinates. We now introduce the concept of weak invertibility, which coincides with classical invertibility when $p=q$ and $\V=\R^q$.

\begin{definition}
\label{wi}
\col{\textbf{(weak invertibility)}} Let $\V \subseteq \R^{q}$ a subspace. A possibly non-square $p \times q$ matrix $M$ is said to be weakly $\V-$invertible if $Ker(M) \cap \V=\{0\}$, i.e. if $Mz \neq 0$ whenever $z \in \V$, $z \neq 0$.
\end{definition}

We now discuss the key technical contribution of theorem \ref{t1}. Let $\mathcal{Z}_n \subseteq \R^{n}$ the subspace of vectors whose coordinates sum to zero. $\mathcal{Z}_n$ is a space of dimension $n-1$ as it is the eigenspace of the zero eigenvalue of the "all ones" $J_n$ matrix. Regarding the eigenvalue analysis needed to show convergence in theorem \ref{t1}, OMWU attacks the antidiagonal part of the Jacobian, in the sense that it aims at modifying the latter so as to get convergence. As stated in \cite{omwu2}, which proves local convergence of OMWU in the convex-concave zero-sum setting: \textit{"The key result of spectral analysis in \cite{omwu} is in Lemma B.6 which states that an skew-symmetric matrix has imaginary eigenvalues.} In our CMWU case, we attack the diagonal part of the Jacobian, which does not exploit skew-symmetry and thus necessitates a different proof technique. Specifically, our technical contribution is to see the problem as the zero eigenvalue perturbation of the real symmetric matrices $J_{|\supp(x^*)|}$, $J_{|\supp(y^*)|}$ \footnote{the support $\supp(z)$ of a vector $z$ is a set which elements are the indexes of the non-zero entries of $z$.} which results in highlighting the interplay between the game equilibrium and the corresponding eigenspaces $\mathcal{Z}_{|\supp(x^*)|}$, $\mathcal{Z}_{|\supp(y^*)|}$ of vectors that sum to zero. We show that for small $h$, the perturbed eigenvalues have strictly positive real part, which makes the Jacobian of the update rule have spectral radius strictly less than one. We obtain a concise theorem condition at the intersection of the game payoff matrix and equilibrium, which explains why CMWU will work well on a large class of games which are not invertible in the classical sense, like rock-paper-scissors. Indeed, for the latter, $Az=0$ for $z \neq 0$ only when $z$ has all equal entries, which does not belong to $\mathcal{Z}_3$. 

\begin{theorem}
\label{t1}
\col{\textbf{(local convergence of CMWU)}} Let $(x^*,y^*)$ a quasi-strict NE of the zero-sum game $A$. Let $A_*$ the payoff matrix whose rows have been restricted to $\supp(x^*)$ and columns to $\supp(y^*)$. Assume that $A_*$ (resp. $A_*^T$) is weakly $\mathcal{Z}_{|\supp(y^*)|}-$invertible (resp. $\mathcal{Z}_{|\supp(x^*)|}-$invertible). Then, for sufficiently small $h$ and $\epsilon=C h^{-\delta}$ satisfying \footnote{we use the convention $|v_A^{-1}|:=+\infty$ if $v_A=0$.} $\epsilon<|v_A^{-1}|$ for some $\delta \in [0,1)$, $C>0$ not depending on $h$, there exists a neighborhood $U(h) \subset \Delta_n \times \Delta_m$ of $(x^*,y^*)$ such that for all $(x^0,y^0) \in U(h)$, the CMWU iterates (\ref{cmwu}) converge to the NE, namely $\lim_{t \to \infty} (\varphi_1(x^t,y^t),\varphi_2(x^t,y^t))=(x^*,y^*)$.
\end{theorem}

\section{Learning to reach Nash Equilibria using Game Signatures}
\label{lg}

We introduce \col{\textbf{\textit{game signatures}}} in section \ref{sig} based on a new game decomposition in theorem \ref{decomp}, which we use as a feature to learn CMWU update's coefficients for general 2P games in section \ref{rl}. This new decomposition of a game given a family of $n$ commutative projectors can be seen as a way to \textit{"play Lego"} with the game and is in our opinion an important observation of this work.

\subsection{Projector-based Game Signatures}
\label{sig}

 We let the pure strategy sets of the two players $\mathcal{X}_1$, $\mathcal{X}_2$ be probability spaces equipped with probability measures $\mu_1$, $\mu_2$. The space of 2P games $\mathcal{G}$ can be identified with the space of pairs of functions $f=(f_1,f_2)$ on $\mathcal{X}_1\times\mathcal{X}_2 \to \mathbb{R}$, where player $j$ receives $f_j(x,y)$ when players play the pure strategies $x$ and $y$.
 
 We seek to answer the question: \textbf{\col{\textit{how to encode a 2P game?}}} In the recent work \cite{chaos}, \textit{trivial} games have been introduced in the bimatrix context: this class of games consists of games which payoff is decoupled. In order to quantify how trivial a game is, authors introduce a triviality metric which consists in computing the closest trivial game to it, by taking the minimum (matrix) distance over all possible trivial games. In a separate work \cite{2pszs}, the concept of \textit{transitive} game was introduced in the case of 2P zero-sum symmetric games, and it was proved that every such game can be decomposed uniquely into a transitive and a cyclic component, i.e. that the space of 2P zero-sum symmetric games is equal to the direct sum of transitive and cyclic games. In the bimatrix case, a transitive game has payoff $A_{ij}=a_i-a_j$. We summarize game definitions from these works in definition \ref{gtypes}.
 
\begin{definition}
\label{gtypes}
\col{\textbf{(Zero-sum, cooperative, symmetric, antisymmetric, trivial, transitive and cyclic games.)}} A game $f=(f_1,f_2)$ is zero-sum if $f_1=-f_2$, cooperative if $f_1=f_2$. If $\mathcal{X}_1=\mathcal{X}_2$ and $f^T(x,y):=f(y,x)$, a game is symmetric if $f_1=f_2^T$, antisymmetric if $f_1=-f_2^T$. A game is trivial \cite{chaos} if $f_j(x,y)=\phi_j(x)+\lambda_j(y)$ for some $\phi_j$, $\lambda_j$. A zero-sum symmetric game is transitive \cite{2pszs} if $f_1(x,y)=\phi(x)-\phi(y)$ for some $\phi$, and cyclic if $\int f_1(x,y)d\mu_2(y)=0$ for all $x \in \mathcal{X}_1$.
\end{definition}

It is then natural to ask the following question: \textit{How to extend the concepts of transitive and cyclic components of a zero-sum symmetric game to any 2P game? Is there a link between transitive games and trivial games?}
 
We provide answers to these questions in proposition \ref{unif} and corollary \ref{d8} by generalizing transitivity, unifying it with triviality, and providing a general game decomposition of 8 components valid for \textit{any} 2P game \footnote{this is valid under the assumption that players have the same pure strategy set, otherwise we get a decomposition into 4 components, as specified in corollary \ref{d8}.}, 2 of which being the zero-sum symmetric transitive and cyclic games in \cite{2pszs}, 4 of which being trivial games of \cite{chaos}, while the 4 other ones constitute a generalization of the concept of cyclicity to any 2P game. Actually we go further: we provide a systematic way to obtain game decompositions of $2^n$ components in direct sum, given $n$ commutative projectors. We build on the key observation in \cite{dgame,proj} that \textit{projection operators}, or \textbf{\col{projectors}}, yield game decompositions. In particular, our unified view shows that every 2P game contains a trivial component, which is easy to compute using $\rho_T$. This represents an alternative to the triviality metric in \cite{chaos}: we can simply define the triviality of a game as the norm of its trivial component. Our toolkit of corollary \ref{d8} yields theorem 1 in \cite{2pszs} as a corollary for free, without needing to resort to combinatorial Hodge theory. We begin with the definition of a projector, central to our analysis.

\begin{definition}
\label{proj}
\col{\textbf{(projector)}} A projector $\rho: \mathcal{G} \to \mathcal{G}$ on $\mathcal{G}$ is a linear operator satisfying $\rho^2=\rho$. In that case we have $\mathcal{G}=K_\rho \oplus R_\rho$ and $K_\rho=R_{\id-\rho}$, where $K_\rho$, $R_\rho$ are the kernel and range of $\rho$.
\end{definition}

Here, $\oplus$ is the direct sum operator, indicating that every element of $\mathcal{G}$ can be decomposed uniquely into an element of $K_{\rho}$ and an element of $R_{\rho}$. In that case, $K_{\rho} \cap R_{\rho} =\{0\}$. Our new decomposition in theorem \ref{decomp} and corollary \ref{d8} is an extension to the case $n>2$ of the well-known lemma \ref{pl}.

\begin{lemma}
\label{pl}
If $p_1$, $p_2$ are commutative projectors on $\mathcal{G}$, then $p_1 p_2$ is a projector and:
$$
\mathcal{G}=(K_{p_1} \cap K_{p_2}) \oplus (K_{p_1} \cap R_{p_2}) \oplus (R_{p_1} \cap K_{p_2}) \oplus (R_{p_1} \cap R_{p_2})\ .
$$
\end{lemma}

\begin{theorem}
\label{decomp}
\col{\textbf{(Game Decomposition)}}
Let $(\rho_i)_{i \in [1,n]}$ a family of $n$ commutative projectors on $\mathcal{G}$. Then we have the canonical direct sum decomposition into $2^n$ components:
\begin{align}
    \mathcal{G}= \bigoplus_{\mathcal{C}_i \in \{K_{\rho_i},R_{\rho_i}\}} \bigcap_{i=1}^n \mathcal{C}_i\ .
\end{align}
\end{theorem}

Theorem \ref{decomp} tells us informally that starting from a game, we can compute its components easily by $n$ compositions of projectors of the form $\rho_i$ or $\id-\rho_j$, where the order doesn't matter because of commutativity.

Let $\mathcal{N}_2f(x,y):= \int f(x,z)d\mu_2(z)$, $\mathcal{N}_1f(x,y):= \int f(z,y)d\mu_1(z)$, $\widehat{\mathcal{N}}:= \mathcal{N}_1+\mathcal{N}_2-\mathcal{N}_1\mathcal{N}_2$. We define the operators $\bm{\col{\rho_{Z}}}$, $\bm{\col{\rho_{S}}}$, $\bm{\col{\rho_{T}}}$ from $\mathcal{G} \to \mathcal{G}$ associated to zero-sum, symmetric and transitive games as:
\begin{align}
\label{refeqop}
\begin{split}
& \bm{\col{\rho_{Z}}}: (f_1,f_2) \to \frac{1}{2}(f_1-f_2, f_2-f_1) \\
&\bm{\col{\rho_{S}}}: (f_1,f_2) \to \frac{1}{2}(f_1+f_2^T, f_2+f_1^T)\\
& \bm{\col{\rho_{T}}}: (f_1,f_2) \to (\widehat{\mathcal{N}}f_1, \widehat{\mathcal{N}}f_2)
\end{split}
\end{align}

In the zero-sum symmetric case, $\rho_T$ reduces to the form $\phi(x)-\phi(y)$ in \cite{2pszs}, but in general it is different (of the form $\phi(x)+\lambda(y)$) and allows to define the transitive/cyclic components of any 2P game given by $\rho_T$ and $\id-\rho_T$. Similarly, $\rho_{Z}$, $\id-\rho_{Z}$, $\rho_{S}$, $\id-\rho_{S}$ compute the zero-sum, cooperative, symmetric and antisymmetric components of a game. \cite{proj} observed that $\rho_{Z}$ and $\rho_{T}$ commute, and use it to obtain decompositions into 2 and 3 components. In our case, we use theorem \ref{decomp} to obtain our new decompositions into $2^n$ components, which allows to generalize and unify \cite{2pszs,chaos}. For example, the cooperative-symmetric-cyclic component of a game $f$ is simply given by $(\id-\rho_Z) \rho_S (\id-\rho_T)f$, and every component is obtained by applying such compositions of projectors.

\begin{proposition}
\label{unif}
\col{\textbf{(generalization of transitivity/cyclicity and unification with triviality)}} If a 2P game $f$ is zero-sum symmetric transitive (resp. cyclic) in the sense of \cite{2pszs}, then $f \in R_{\rho_{T}}$ (resp. $f \in K_{\rho_{T}}=R_{\id-\rho_T}$). Further, a 2P game is trivial in the sense of \cite{chaos} if and only if it is transitive, i.e. the class of trivial games coincides with $R_{\rho_{T}}$.
\end{proposition}

\begin{corollary}
\label{d8}
\col{\textbf{(Canonical decomposition of 2P games)}}
$\rho_{Z}$, $\rho_{T}$ are commutative projectors. Consequently, by theorem \ref{decomp}, every 2P game can be decomposed into the direct sum of 4 components: \col{\textbf{(ZT)}} zero-sum-transitive, \col{\textbf{(ZCy)}} zero-sum-cyclic, \col{\textbf{(CT)}} cooperative-transitive, \col{\textbf{(CCy)}} cooperative-cyclic. Further, if $\mathcal{X}_1=\mathcal{X}_2$, then $\rho_{S}$ is a well-defined projector and commutes with $\rho_{Z}$, $\rho_{T}$. Consequently, any such game can be decomposed into the direct sum of 8 components: \col{\textbf{(ZST)}} zero-sum-symmetric-transitive, \col{\textbf{(ZSCy)}} zero-sum-symmetric-cyclic, \col{\textbf{(ZAT)}} zero-sum-antisymmetric-transitive, \col{\textbf{(ZACy)}} zero-sum-antisymmetric-cyclic, 
\col{\textbf{(CST)}} cooperative-symmetric-transitive, \col{\textbf{(CSCy)}} cooperative-symmetric-cyclic, \col{\textbf{(CAT)}} cooperative-antisymmetric-transitive, \col{\textbf{(CACy)}} cooperative-antisymmetric-cyclic.
\end{corollary}

Finally, we define the concept of game norm, which is used to construct the game signature based on theorem \ref{decomp} and corollary \ref{d8}. 

\begin{definition}
\label{gnorm}
\col{\textbf{(Game norm)}} Let $f=(f_1,f_2) \in \mathcal{G}$ a game, and $||\cdot||$ a norm on the space of functions $\mathcal{X}_1\times\mathcal{X}_2 \to \mathbb{R}$. The game norm $||f||$ is defined as $||f||:=\frac{1}{2}(||f_1||+||f_2||)$.
\end{definition}

\begin{definition}
\label{gsig}
\col{\textbf{(Game signature)}} Let $(\rho_i)_{i \in [1,n]}$ a family of $n$ commutative projectors. The game signature is defined as the vector of size $2^n$ consisting of the norms of the $2^n$ game components associated to the decomposition of theorem \ref{decomp}.
\end{definition}

Let us look at how to compute the game components and signature in the bimatrix case with payoff matrices $(A,B)$. We have from (\ref{refeqop}), $\rho_Z(A,B)=\frac{1}{2}(A-B,B-A)$, $\rho_S(A,B)=\frac{1}{2}(A+B^T,B+A^T)$, $\rho_T(A,B)=(\widehat{A}^{(1)}+\widehat{A}^{(2)}-\widehat{A},
\widehat{B}^{(1)}+\widehat{B}^{(2)}-\widehat{B})$, where $\widehat{A}^{(1)}$, $\widehat{A}^{(2)}$, $\widehat{A}$ are the matrices with entries $\widehat{A}^{(1)}_{ij}:=\frac{1}{m}\sum_{j=1}^m A_{ij}$, $\widehat{A}^{(2)}_{ij}:=\frac{1}{n}\sum_{i=1}^n A_{ij}$, $\widehat{A}_{ij}:=\frac{1}{mn}\sum_{i=1}^n \sum_{j=1}^m A_{ij}$. The projectors associated respectively to cooperative, antisymmetric and cyclic games are $\rho_C(A,B):=[\id-\rho_Z](A,B)$, $ \rho_A(A,B):=[\id-\rho_S](A,B)$, $\rho_{Cy}(A,B):=[\id-\rho_T](A,B)$. With these definitions, the 8 components of corollary \ref{d8} are simply computed by composition of above projectors, namely $\rho_Z\rho_S\rho_T$, $\rho_Z\rho_S\rho_{Cy}$, $\rho_Z\rho_{A}\rho_T$, $\rho_Z\rho_{A}\rho_{Cy}$, $\rho_{C}\rho_S\rho_T$, $\rho_{C}\rho_S\rho_{Cy}$, $\rho_{C}\rho_{A}\rho_T$, $\rho_{C}\rho_{A}\rho_{Cy}$. Following definitions 6 and 7, the game signature is taken to be the vector of size 8 containing the norms $\frac{1}{2}(||A_i||+||B_i||)$ of the components $(A_i,B_i)_{i \in [1,8]}$, where the matrix norm is chosen to be the $L_2$ norm. We further divide these 8 norms by their overall sum so as to interpret them as weights.

\subsection{RL formulation of signature-based learning in games}
\label{rl}

\textbf{Game signature-based learning.} We use our previous analysis to extend to games the framework developed in \cite{l2o}, where a RL policy $\pi$ outputs the step $\Delta x$ in the function optimization update $x^{t+1} \leftarrow x^t +\Delta x$ as a function of current and past "trajectory information" such as gradient and function value at previous points. Since finding one NE of a bimatrix game is PPAD-hard \cite{DBLP:journals/jacm/ChenDT09,DBLP:journals/siamcomp/DaskalakisGP09}, we cannot expect a universal learnt policy to be polynomial-time for unrestricted inputs. However, worst-case "PPAD-hard instances" are very pathological so there is every possibility for good performance over many reasonable distributions over games. 

There are two fundamental differences between games and single functions: \textbf{(i)} the solution concept is not the same: in games we are not trying to minimize a function, but to find a NE \footnote{other solution concepts exist, such as correlated equilibria.}. \textbf{(ii)} the update rule should depend on the nature of the game, as discussed in section \ref{sec1} on the mirror behavior of CMWU and MWU. \textbf{(i)} can be dealt with using the concept of \textbf{\col{best response gaps}} in definition \ref{ng}, whose sum $\delta$ is also called "Nash convergence metric" in \cite{nd}. Indeed, by definition, $(x^*,y^*)$ is a NE if and only if $\delta_1(x^*,y^*)=\delta_2(x^*,y^*)=0$. One of the main ideas in this work is to give an answer to \textbf{(ii)} using our concept of game signature in definition \ref{gsig}, which will allow to learn a range of behaviors across the spectrum of games. In contrast to classical optimization where a suite of convergent algorithms exist, there is no algorithm that displays last-iterate convergence to NE in the case of general 2P games. In the following we consider bimatrix games $(A,B)$ of size $n \times m$, where $x \in \Delta_n$, $y \in \Delta_m$.

\begin{definition}
\label{ng}
\col{\textbf{(best response gap)}} The best response gaps of players 1 and 2 are defined as:
\begin{align}
\begin{split}
\delta_1(x,y):= \max_{i \in [1,n]} [Ay]_i-x^TAy \\
\delta_2(x,y):=\max_{j \in [1,m]} [x^TB]_j-x^TBy\\
\delta(x,y):=\delta_1(x,y)+\delta_2(x,y)
\end{split}
\end{align}
\end{definition}

\textbf{RL state, actions and rewards.} We frame our learning problem in the standard single-agent RL paradigm. We consider the modified CMWU update (\ref{cmwu}) where each player $k \in \{1,2\}$ is allowed to have its own gradient and Hessian coefficients $h_k$, $\epsilon_k$, on which we relax the positivity constraint. Note that in the bimatrix case $(A,B)$, simplex-Hessians become $H_x:=A^TD_xA$, $H_y:=BD_yB^T$, see e.g. \cite{dmg} for the unconstrained case. We allow $h_k=h_k(x^t,y^t)$, $\epsilon_k=\epsilon_k(x^t,y^t)$ to depend on current trajectory location, importantly we do not allow these parameters to depend explicitly on time $t$. At stage $t$, the \col{\textbf{state}} $s_t$ of the RL policy consists of the game signature plus trajectory information $\delta_1(x^t,y^t)$, $\delta_2(x^t,y^t)$, $Ay^t$, $B^Tx^t$, $H_{x^t}$, $H_{y^t}$, $x^{t,T}Ay^t$, $x^{t,T}By^t$. Its \col{\textbf{action}} $a_t \sim \pi(\cdot|s_t) \in \mathbb{R}^4$ consists of $a_t=(h_1,h_2,h_1\epsilon_1,h_2\epsilon_2)$. We optimize the following objective over a trajectory of length $\tau$:
\begin{align}
    \max_\pi \EE_\pi \sum_{t=0}^{\tau-1} R(s_t,a_t), \hspace{1mm} R(s_t,a_t):=-\frac{\delta(x^{t+1},y^{t+1})}{ \delta(x^0,y^0)}
\end{align}
where the expectation is taken over the initial location $(x^0,y^0)$, the payoff matrices $A$, $B$ of the two players (randomly sampled at the beginning of the episode), and the stochasticity of $\pi$. We formulated our problem using RL, as RL policies tend to be more stable and overcome the problem of compounding errors, where errors accumulated on a trajectory could lead to divergence \cite{l2o}. This allowed us in practice to train on long trajectories $\tau=1000$. However, we believe that it would also be possible to use LSTM-based approaches \cite{l2l,blo}, or \cite{l20mm} that introduces an architecture to solve unconstrained min-max problems, and possibly output directly the updates $\Delta x$, $\Delta y$ instead of the update rule's coefficients. We leave this topic for future research. On the other hand, outputing the coefficients allows us to obtain nice relationships in section \ref{exp} such as high discrepancy between gradient and Hessian coefficients on cyclic games, as opposed to transitive games where the gradient coefficient is high, and build on the mirror relationship seen in \cite{chaos}, which another way to reformulate is that the coefficients should be chosen as a function of the game type.

\section{Experiments}
\label{exp}

\textbf{Zero-sum bimatrix games.} We first evaluate in table \ref{tab2} the performance of CMWU against known algorithms on 200 randomly sampled bimatrix zero-sum games in dimension $n=m=$ 25, 50, 75, 100, where each matrix entry, as well as the initial location $x^0$, $y^0$ are sampled independently $\mathcal{U}[-1,1]$, where $\mathcal{U}$ is the uniform distribution \footnote{we apply softmax to $x^0$, $y^0$.}. To quantify speed of convergence to NE and treat different seeds equally with respect to their random starting point, we consider the convergence metric $\beta_\tau:=100 \cdot \tau^{-1}\sum_{t=1}^\tau \frac{\delta(x^t,y^t)}{ \delta(x^0,y^0)}$ which quantifies the trajectory-averaged best response gap relative to the starting point, for a given iteration budget $\tau$, in percentage points. OMWU, OGDA and OMD all have one parameter, the learning rate $h$, whereas CMWU has in addition the Hessian coefficient $\epsilon$. We run experiments for a range of constant learning rates $h$, and pick the best for each algorithm and each game size. The complete set of results for all values of $h$ and score standard deviations are reported in the appendix, together with background on these algorithms, see also \cite{omwu}, \cite{merti}, \cite{ogda}. We find that OGDA performs best, but CMWU performs competitively and does well on a wider range of learning rates. Note that OMWU, CMWU and OMD all use softmax projection, but OGDA requires an orthogonal projection onto the unit simplex with respect to the $L_2$ norm, which requires vector sorting.

\begin{table}[!th]
\caption{Zero-sum bimatrix games: average score $\beta_{500}$ over 200 seeds, for various game sizes $n$. Best learning rate $h$ for each algorithm. $\epsilon=0.25 \cdot h^{-1}n$.}
\label{tab2}
\begin{center}
  \begin{tabular}{ccccc}
    \hline
    $n$ & CMWU & OGDA & OMD & OMWU \\
    25 & 8& \textbf{2}& 25& 26 \\
    50 & 7& \textbf{3}&23 &24  \\
    75 & 8&\textbf{2} &21 &24  \\
    100 & 8 & \textbf{2} & 17 &24  \\
    \hline
  \end{tabular}\end{center}
\end{table}
\vspace{3mm}
\textbf{Learning last-iterate convergence on general games.} We consider square games of size $n=m=10$. During training, for each RL episode of length $\tau=1000$, we sample randomly $A$, $B$, $(x^0,y^0)$ as previously discussed. We compute the 8 pure components $[(A_j,B_j)]_{j \in [1,8]}$ of corollary \ref{d8}, and consider a mixture $(\sum_{k=1}^3 w_k A_{j_k}, \sum_{k=1}^3 w_k B_{j_k})$ of 3 such components $j_k$ chosen at random ($j_k$'s are not necessarily distinct), with respective weights $w_k$ chosen at random in the probability simplex. We then apply the RL-based methodology in section \ref{rl} to learn the algorithm coefficients along trajectories.

We consider 3 versions of each algorithm: \textbf{\textit{Base}}, \textbf{\textit{Partial-RL}} and \textbf{\textit{Full-RL}}, all trained with the methodology described above. The RL policy in the \textbf{\textit{Base}} method learns the coefficients of the vanilla versions of the algorithms ($h$ for all algorithms, $\epsilon$ in addition for CMWU), and has its state $s_t$ masked, i.e. it doesn't know about trajectory information nor the game signature (cf. section \ref{rl}), and therefore outputs constant coefficients (chosen however to optimize performance during training across the spectrum of games). The \textbf{\textit{Partial-RL}} method is an improvement over the Base method in that the policy knows about trajectory information and the game signature. The \textbf{\textit{Full-RL}} method learns 4 coefficients for each algorithm. For CMWU, we learn $(h_1, h_2, \epsilon_1,\epsilon_2)$ as mentioned in section \ref{rl}. For OGDA and OMD, we learn separate coefficients for both players and both the mirror and regular steps. For OMWU, we learn different coefficients for the current and previous gradients $\nabla_t$, $\nabla_{t-1}$. 

We train policies with PPO \cite{ppo}. At test time, we evaluate on mixtures of $k=1$ to 4 pure components as reported in table \ref{tab1}, where the acronyms used are defined in corollary \ref{d8}, over 200 seeds (iteration budget $\tau$ is detailed in appendix). We construct these $k-$mixtures as previously discussed. Note that for clarity of presentation, we have grouped the symmetric and antisymmetric versions of the games on the same row, i.e. the row ZT will contain mixtures of ZST and ZAT, the row ZT + CCy mixtures of ZST, ZAT, CSCy, CACy. We find in table \ref{tab1} that RL-based methods perform well across the spectrum of games compared to the base counterparts, indicating that the policy is able to exploit trajectory information as well as the game signature. CMWU displays slightly best performance across all $k$-mixtures, notably on the challenging mixture of zero-sum and cooperative cyclic components ZCy + CCy, where other methods fail. There, we learn to apply low, sometimes negative learning rate, as well as high Hessian coefficient in order to "tame" the cyclicity. In figure \ref{ff3}, we illustrate this phenomenon by displaying last-iterate convergence $\frac{\delta(x^t,y^t)}{ \delta(x^0,y^0)}$ as a function of time $t$. To further confirm that the policy is able to exploit the signature, we conducted a small ablation study: we ran the Partial-RL CMWU method of table \ref{tab1} by zeroing-out the game signature, so that the policy input contains all trajectory information except from the signature. We get results in line with the Base method, where no information is known, which shows that the signature is crucial: specifically, ZCy + CCy and ZT + ZCy + CCy respectively decrease from 8 to 96, 19 to 70, hence we learn how to converge in these games by exploiting the signature.

\begin{figure}[th]
 \centering
  \centerline{\includegraphics[scale=0.5]{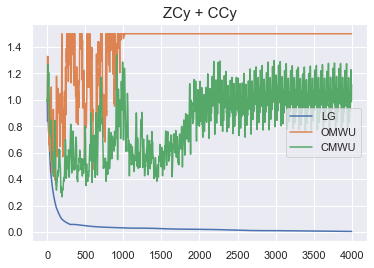}}
  \caption{\textcolor{darktangerine}{ZCy} + \textcolor{deepcerise}{CCy}: last-iterate convergence $t \to \frac{\delta(x^t,y^t)}{ \delta(x^0,y^0)}$. Full-RL CMWU learns to use negative gradient coefficient $h_k$ and high Hessian coefficient $|h_k\epsilon_k|$. OMWU and CMWU diverge.}
  \label{ff3}
\end{figure}

In figure \ref{ff2} we plot for Full-RL CMWU the mean value and pairwise correlations among the learnt coefficients across the 8 pure game components \footnote{we compute the correlation between the coefficients time-series for each seed, and average over seeds.}. We see that \textbf{(1)} for transitive games, our method learns to apply higher gradient coefficients $h_k$: this matches intuition since these games have decoupled payoff, hence players can act more intensely on their gradients. For cyclic games, it applies higher absolute Hessian to counter the cyclicity. \textbf{(2)} We see a \textbf{\col{mirror effect}} between zero-sum and cooperative games as suggested in \cite{chaos}: Hessian coefficients have opposite sign, which also reflects on the correlations. \textbf{(3)} The performance and learned coefficients for the symmetric and antisymmetric parts of a given game type are almost identical, which seems to indicate that this feature has little importance. \textbf{(4)} The correlation $\corr(\delta_1-\delta_2,h_1-h_2)$ is always significantly positive, $74\%$ in average. This hints at a possible "WoLF effect" learnt by the policy, similar to \cite{wolf}: when the best response gap of player 1 is low compared to player 2 (player 1 "winning"), then player 1 becomes more careful. We leave this interesting observation for future work.

In figure \ref{fff2} we compute, for the full-RL CMWU method, the average standardized per-episode-trajectory, i.e. "shape", of coefficients of gradient $G_k=h_k$ and Hessian $H_k=-h_k \epsilon_k $ of players $k=1,2$ across game types. We obtain these average trajectories as follows: for each seed, we consider the standardized time-series of coefficients \footnote{substract the mean, divide by standard deviation.}, and then average the latter over seeds. We see that typically, there is a transient regime where the coefficients change significantly, followed by a stationary regime where they stabilize (neighborhood of the Nash). In addition to the mirror effect between zero-sum and cooperative games discussed earlier in figure \ref{ff2}, we see a another such mirror effect where in the zero-sum case, all coefficients are higher in the transient regime than in the stationary regime, whereas in the cooperative case, both Hessian coefficients adopt an opposite behavior where they are lower in the transient regime.

\begin{figure}[ht]
  \centering
  \centerline{\includegraphics[width=1.\columnwidth]{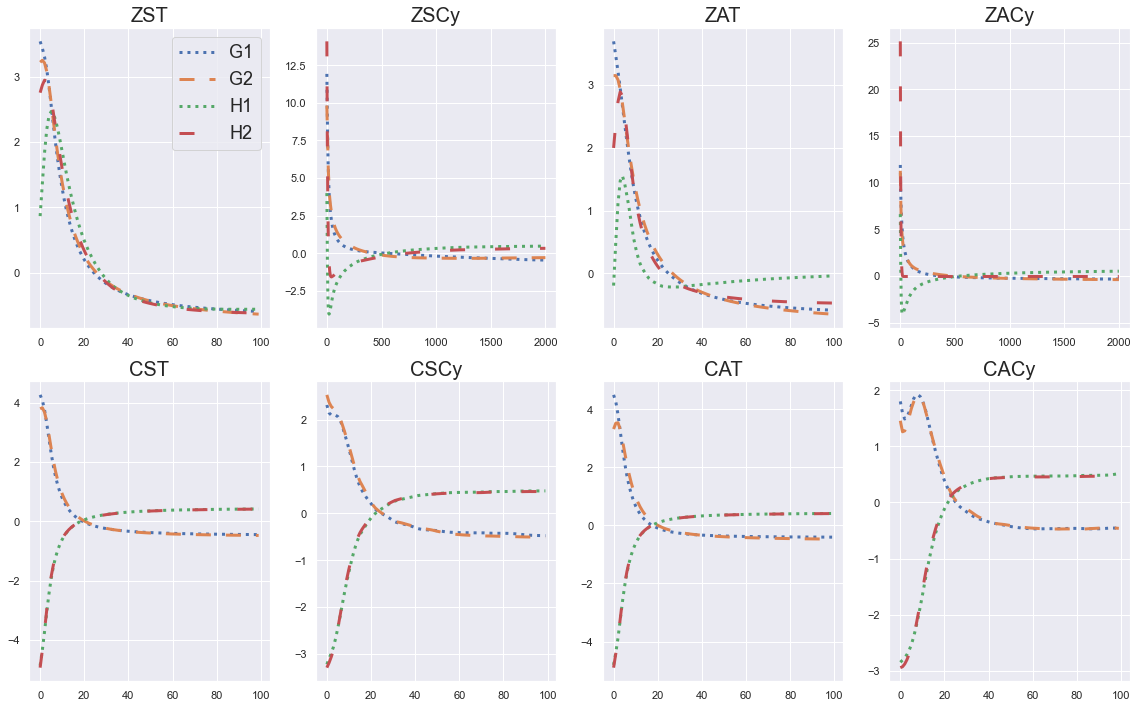}}
  \caption{Average standardized per-episode-trajectory ("shape") of coefficients of gradient $G_k=h_k$ and Hessian $H_k=-h_k \epsilon_k $ of players $k=1,2$ across game types, as a function of time $t$.}
  \label{fff2}
\end{figure}

 In figure \ref{fff21}, we used the standard KernelExplainer tool associated to the SHAP algorithm \cite{shap} to look at the contributions of 4 input feature groups ("BR": best response gaps $\delta_k$, "GRAD": gradient, "HESS": Hessian, "PAY": payoff value) to the gradient and Hessian coefficients $h_k$, $-h_k \epsilon_k$, which constitute the RL policy output. We see that in the transitive case, the gradient is the most important, but for cyclic games, we get a more varied behavior. For each of the 4 input groups, the percentage contribution displayed is the sum of the absolute SHAP contributions of its individual components: for example, for the best response input "BR", we sum the absolute contributions of the two player's best responses $\delta_1$ and $\delta_2$. We normalize SHAP contributions such that for each game type, they sum to one. Further details are provided in appendix, where we also display percentage contributions for the "signed" case where for each one of the 4 feature groups, we sum over the signed SHAP contributions of its components, rather than absolute.
 
 \begin{figure}[ht]
  \centering
  \centerline{\includegraphics[width=1.\columnwidth]{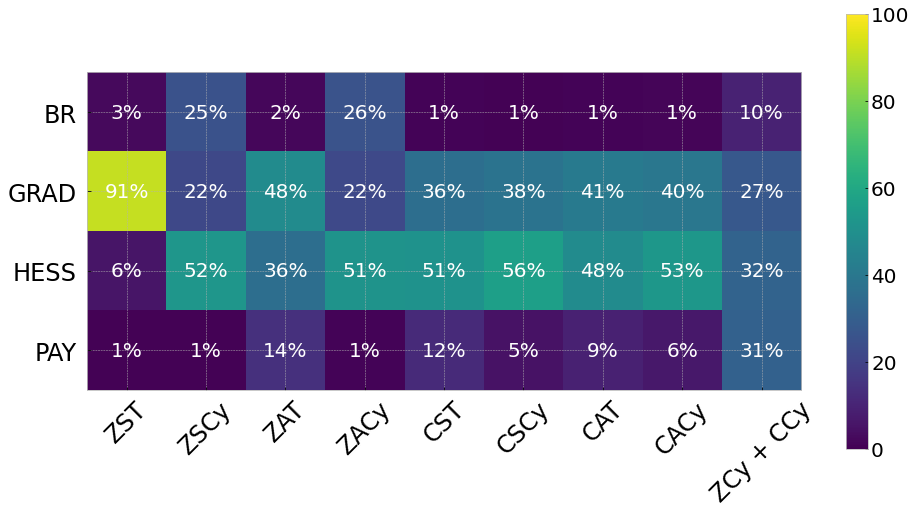}}
  \centerline{\includegraphics[width=1.\columnwidth]{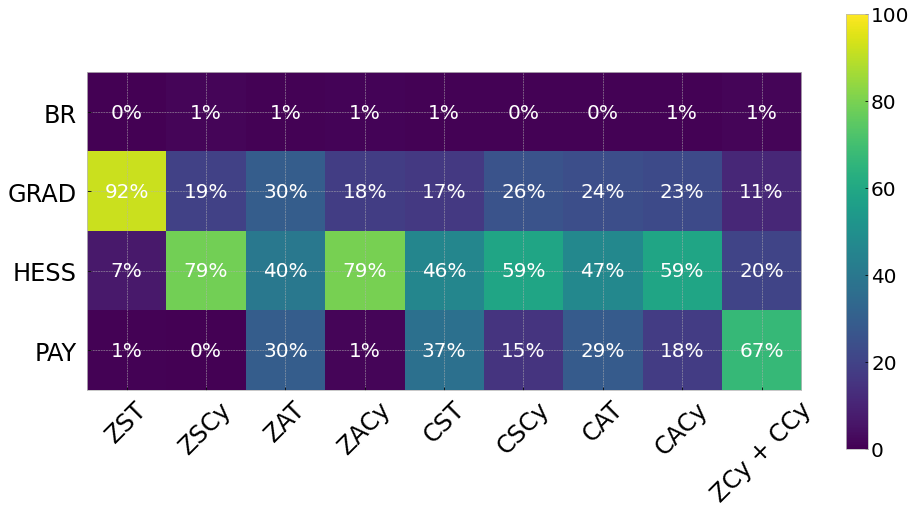}}
  \caption{Percentage contribution towards absolute SHAP importance aggregated across absolute values with respect to \textbf{(top)} gradient learning rates $G_k=h_k$ and \textbf{(bottom)} Hessian learning rates $H_k=-h_k \epsilon_k$ for players $k$ for the four input groups and eight unique game types. Columns are normalized to sum to one.}
  \label{fff21}
\end{figure}

\section{Conclusion and Future Research}
\label{secc}

We introduced a new algorithm, CMWU, together with a new framework for learning NE based on the nature of the game, materialized by our game signature and constructed using new game decompositions that generalize and unify game concepts studied in the recent works. We believe this work asks new research questions, among which we suggest \textbf{(1)} Extension of CMWU from zero-sum bimatrix to the convex-concave setting, as was done for OMWU \cite{omwu2}, and investigation of global convergence properties of CMWU. \textbf{(2)} Allowing the neural network to output the update steps $\Delta x$, $\Delta y$ directly as done in \cite{l2l,blo,l20mm}, as opposed to the update rule's coefficients. \textbf{(3)} Finding other relevant projectors in the space of 2P games. Can the game signature be learnt from experience? If yes, can the learnt features be mapped to new projectors?

\clearpage

\begin{figure*}[th]
\centering
  \centerline{\includegraphics[width=\textwidth]{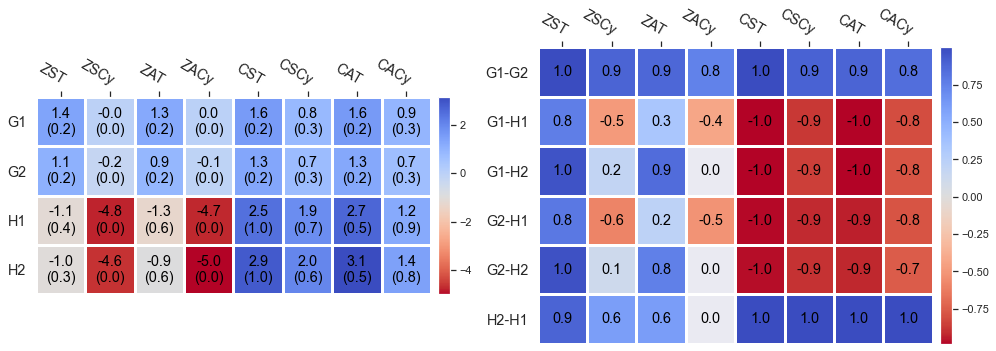}}
  \caption{Full-RL CMWU coefficients of gradient $G_k=h_k$ and Hessian $H_k=-h_k \epsilon_k $ of players $k=1,2$ across 8 pure game components. \textbf{(left)} Mean value. \textbf{(right)} Pairwise correlations. StDev in brackets.}
  \label{ff2}
\end{figure*}

\begin{table*}[!ht]
\caption{Learning Games for various algorithms and methods, cf. details in text. OMWU (\textbf{O}), CMWU (\textbf{C}), OGDA (\textbf{G}),  OMD (\textbf{M}). Full-RL (\textbf{F}), Partial-RL (\textbf{P}), Base (\textbf{B}). Average score $\beta_\tau$ over 200 seeds.}
\label{tab1}
\begin{center}
\begin{tabular}{lllllllllllll}
\hline
\makecell{Game Type} & \makecell{C-F} & \makecell{G-F} & \makecell{M-F} & \makecell{O-F} & \makecell{C-P} & \makecell{G-P} & \makecell{M-P} & \makecell{O-P} & \makecell{C-B} & \makecell{G-B} & \makecell{M-B} & \makecell{O-B}\\ 
\hline
\textcolor{hanpurple}{ZT}
&6&93&5&5&8&3&4&21&7&\textbf{1}&3&4  \\  
\textcolor{darktangerine}{ZCy}
& 4& 254 & 19 & 16 & \textbf{3} & 79 & \textbf{3} & 73 & 7 & 437 &\textbf{3}& 4  \\ 
 \textcolor{darkpastelgreen}{CT}
 & 4 & 119 & 5 & 5 & 5 & \textbf{1} & 2 & 14 & 9 & \textbf{1} & 2 & 4  \\  
\textcolor{deepcerise}{CCy}
&4&791&4&4&6&1&3&11&21&\textbf{0}&3&3 \\ 
\hline
\textbf{Avg. Mixtures of 1}
& 5 & 314 & 8 & 8 & 6 & 21 & \textbf{3} & 30 & 11 & 110 & \textbf{3} & 4 \\
\hline
\textcolor{hanpurple}{ZT} + \textcolor{darktangerine}{ZCy}
&\textbf{7}&89&10&17&\textbf{7}&16&42&40&29&122&53&47\\  
\textcolor{hanpurple}{ZT} + \textcolor{darkpastelgreen}{CT}
&3&91&3&3&4&1&2&11&6&\textbf{0}&1&2 \\ 
\textcolor{hanpurple}{ZT} + \textcolor{deepcerise}{CCy}
&6&289&4&4&7&\textbf{1}&3&13&16&\textbf{1}&3&4\\ 
 \textcolor{darktangerine}{ZCy} + \textcolor{darkpastelgreen}{CT}
 &15&101&\textbf{10}&15&16&16&26&39&11&147&25&16\\ 
 \textcolor{darktangerine}{ZCy} + \textcolor{deepcerise}{CCy}
 &10&322&30&58&\textbf{8}&25&120&104&116&137&123&119\\  
\textcolor{darkpastelgreen}{CT} + \textcolor{deepcerise}{CCy}
&5&275&5&5&6&\textbf{1}&3&13&12&\textbf{1}&3&4\\ 
\hline
\textbf{Avg. Mixtures of 2}
& \textbf{8} & 195 & 10 & 17 & \textbf{8} & 10 & 33 & 37 & 32 & 68 & 35 & 32\\
\hline
\textcolor{hanpurple}{ZT} + \textcolor{darktangerine}{ZCy} + \textcolor{darkpastelgreen}{CT}
& 12&39&9&18&13&\textbf{3}&38&35&22&40&39&32\\ 
\textcolor{hanpurple}{ZT} + \textcolor{darktangerine}{ZCy} + \textcolor{deepcerise}{CCy}
& 19&114&23&36&19&\textbf{13}&94&66&69&41&93&87\\ 
\textcolor{hanpurple}{ZT} + \textcolor{darkpastelgreen}{CT} + \textcolor{deepcerise}{CCy}
& 5&104&5&5&7&\textbf{1}&2&14&11&\textbf{1}&2&3\\ 
\textcolor{darktangerine}{ZCy} + \textcolor{darkpastelgreen}{CT} + \textcolor{deepcerise}{CCy}
& 20&122&18&32&28&\textbf{13}&84&61&58&33&78&72\\
\textcolor{hanpurple}{ZT} + \textcolor{darktangerine}{ZCy} + \textcolor{darkpastelgreen}{CT} + \textcolor{deepcerise}{CCy}
& 15&49&12&18&18&\textbf{5}&40&32&36&9&38&34\\
\hline
\textbf{Avg. Mixtures of 3-4}
& 14 & 86 & 13 & 22 & 17 & \textbf{7}&52 &42 &39 &25 & 50 &46\\
\textbf{Avg. $k$-Mixtures}
& \textbf{9}& 198 & 10 & 16 & 10 &13 & 29 & 36& 27& 68& 29 &27\\
\hline
\end{tabular}\end{center}
\end{table*}


\clearpage
\clearpage

\bibliography{neurips_2021}
\bibliographystyle{icml2022}

\newpage
\appendix
\onecolumn

\section*{Disclaimer}
This paper was prepared for information purposes by the Artificial Intelligence Research group of JPMorgan Chase \& Co and its affiliates (“JP Morgan”), and is not a product of the Research Department of JP Morgan. JP Morgan makes no representation and warranty whatsoever and disclaims all liability, for the completeness, accuracy or reliability of the information contained herein. This document is not intended as investment research or investment advice, or a recommendation, offer or solicitation for the purchase or sale of any security, financial instrument, financial product or service, or to be used in any way for evaluating the merits of participating in any transaction, and shall not constitute a solicitation under any jurisdiction or to any person, if such solicitation under such jurisdiction or to such person would be unlawful. \textsuperscript{\textcopyright} 2022 JPMorgan Chase \& Co. All rights reserved. 

\section{Background}
\label{seca1}
\textbf{Consensus optimization in the unconstrained case.} In the unconstrained case where players' pure strategy sets are given by $\mathcal{X}_1=\mathbb{R}^n$, $\mathcal{X}_2=\mathbb{R}^m$, and their payoffs $f_1$, $f_2$ $:\mathcal{X}_1 \times \mathcal{X}_2 \to \mathbb{R}$, we define the game gradient and Jacobian as follows \cite{dmg2}:
\begin{align*}
v:=\begin{pmatrix}\nabla_x f_1\\ \nabla_y f_2 \end{pmatrix}, \hspace{7mm}
H:=\begin{pmatrix}
\nabla^2_x f_1 & \nabla^2_{x,y} f_1 \\
\nabla^2_{x,y} f_2 & \nabla^2_{y} f_2
\end{pmatrix}
\end{align*}
The consensus optimization update rule takes the following form, for $x \in \mathcal{X}_1$, $y \in \mathcal{X}_2$ \cite{dmg2,co}:
\begin{align*}
\begin{pmatrix}
x \\
y
\end{pmatrix} \leftarrow 
\begin{pmatrix}
x \\
y
\end{pmatrix} + h v - h\epsilon H^Tv
\end{align*}
where $h$, $\epsilon$ are some parameters. In particular in the (unconstrained) bilinear case $f_1(x,y)=x^TAy$, $f_2(x,y)=x^TBy$ where $A$, $B$ are some matrices of size $n \times m$, we get:
\begin{align*}
\begin{pmatrix}
x \\
y
\end{pmatrix} \leftarrow 
\begin{pmatrix}
x \\
y
\end{pmatrix} + h 
\begin{pmatrix}
Ay \\
B^Tx
\end{pmatrix} 
- h\epsilon 
\begin{pmatrix}
0 & A  \\
B^T & 0
\end{pmatrix}^T \begin{pmatrix}
Ay \\
B^Tx
\end{pmatrix} 
= \begin{pmatrix}
x +hAy-h\epsilon BB^Tx\\
y + hB^Tx-h\epsilon A^TAy
\end{pmatrix}
\end{align*}

\textbf{OMWU.} For a constrained bimatrix game $(A,B)$ associated to players payoffs $x^TAy$, $x^TBy$, where the mixed strategies in the probability simplex $x \in \Delta_n$, $y \in \Delta_m$, the OMWU update rule \cite{chaos,omwu} is given by, for $i \in [1,n]$, $j \in [1,m]$:
\begin{align*}
    \begin{split}
    &x_i^{t+1}=  \frac{x_i^{t}\exp(2h[Ay^t]_i -h[Ay^{t-1}]_i)}{\sum_{k=1}^n x^t_k\exp(2h[Ay^t]_k -h[Ay^{t-1}]_k)}\\
    &y_j^{t+1}=  \frac{y_j^{t}\exp(2h[B^Tx^t]_j -h[B^Tx^{t-1}]_j)}{\sum_{k=1}^m y_k^{t}\exp(2h[B^Tx^t]_j -h[B^Tx^{t-1}]_j)}\ .
    \end{split}
\end{align*}

\textbf{OMD.} Using the same notations as OMWU, the OMD update is given by \cite{merti}:
\begin{align*}
    \begin{split}
    &\widehat{x}_i^{t}=  \frac{x_i^{t}\exp(h[Ay^t]_i)}{\sum_{k=1}^n x^t_k\exp(h[Ay^t]_k)}, \hspace{3mm}
    \widehat{y}_j^{t}=  \frac{y_j^{t}\exp(h[B^Tx^t]_j )}{\sum_{k=1}^m y_k^{t}\exp(h[B^Tx^t]_j )}\\
    &x_i^{t+1}=  \frac{x_i^{t}\exp(h[A\widehat{y}^t]_i)}{\sum_{k=1}^n x^t_k\exp(h[A\widehat{y}^t]_k)}, \hspace{3mm}
    y_j^{t+1}=  \frac{y_j^{t}\exp(h[B^T\widehat{x}^t]_j )}{\sum_{k=1}^m y_k^{t}\exp(h[B^T\widehat{x}^t]_j )}
    \end{split}
\end{align*}

\textbf{OGDA.} Let $\widehat{x}^{0}:=x^{0}$, $\widehat{y}^{0}:=y^{0}$, and $\Pi$ be the orthogonal projection onto the unit simplex with respect to the $L_2$ norm, i.e. $\Pi(x)=\argmin_{y \in \Delta_n} ||x-y||_2$. Using the same notations as OMWU, the OGDA update is given by \cite{ogda}: 
$$
x^{t+1} = \Pi \left[ \widehat{x}^{t}+h Ay^t \right], \hspace{3mm} y^{t+1} = \Pi \left[ \widehat{y}^{t}+h B^Tx^t \right]
$$
$$
\widehat{x}^{t+1} = \Pi \left[ \widehat{x}^{t}+h Ay^{t+1} \right], \hspace{3mm} \widehat{y}^{t+1} = \Pi \left[ \widehat{y}^{t}+h B^Tx^{t+1} \right]
$$

\section{Proofs}
\label{seca2}

\begin{proposition}
\label{ostro}
(Proposition 2.1 in \cite{omwu}). Let $\phi:\Delta \to \Delta$, with $\Delta:=\Delta_n \times \Delta_m$, and consider the recurrence relation $z^{t+1}=\phi(z^t)$. If the Jacobian of the update rule $\phi$ at a fixed point $z^*$ (i.e. $\phi(z^*)=z^*$) has spectral radius strictly less than one, then there exists a neighborhood $U$ around $z^*$ such that for all $z \in U$, the dynamics starting from $z$ converge to $z^*$, i.e. $\lim_{t \to \infty}\phi(z^t)=z^*$, $z^0=z$.
\end{proposition}

\begin{lemma}
\label{lemeig}
Let $S$ a real symmetric matrix for which 0 is an eigenvalue, and let $z$ a unit $L_2-$norm eigenvector associated to it. Let $\lambda_\epsilon$ and $z+z_\epsilon$ the perturbed eigenpair (possibly complex) associated to $S+\epsilon M$, where $M$ is a real matrix. Then we have, as $\epsilon \to 0$:
$$
\lambda_\epsilon = \epsilon z^TM z + o(\epsilon)
$$
\end{lemma}
\begin{proof}
The result is known in matrix perturbation theory, we recap the proof here for convenience. We have by definition $(S+\epsilon M)(z+z_\epsilon)=\lambda_\epsilon (z+z_\epsilon)$, hence:
\begin{align*}
\epsilon M z + S z_\epsilon + \epsilon M z_\epsilon = \lambda_\epsilon z + \lambda_\epsilon z_\epsilon 
\end{align*}
We have $z_\epsilon=o(1)$ as $\epsilon \to 0$ by continuity of the eigenvectors and eigenvalues with respect to $\epsilon$ \cite{bhatia}, hence $\epsilon M z_\epsilon=o(\epsilon)$. By symmetry of $S$ we have $z^TS=(Sz)^T=0$, so multiplying by $z^T$ on the left yields:
\begin{align*}
 & \epsilon z^TM z  = \lambda_\epsilon  + \lambda_\epsilon \left<z_\epsilon,z\right> + o(\epsilon)\\
 \Rightarrow & \hspace{2mm}\frac{\lambda_\epsilon}{\epsilon} =  \frac{1}{1+o(1)} z^TM z + o(1) 
\end{align*}
Taking the limit yields $\lim_{\epsilon \to 0} \frac{\lambda_\epsilon}{\epsilon} = z^TM z$, hence $\lambda_\epsilon=\epsilon z^TM z +o(\epsilon)$. 
\end{proof}

\begin{theorem}
\col{\textbf{(local convergence of CMWU)}} Let $(x^*,y^*)$ a quasi-strict NE of the zero-sum game $A$. Let $A_*$ the payoff matrix whose rows have been restricted to $\supp(x^*)$ and columns to $\supp(y^*)$. Assume that $A_*$ (resp. $A_*^T$) is weakly $\mathcal{Z}_{|\supp(y^*)|}-$invertible (resp. $\mathcal{Z}_{|\supp(x^*)|}-$invertible). Then, for sufficiently small $h$ and $\epsilon=C h^{-\delta}$ satisfying \footnote{we use the convention $|v_A^{-1}|:=+\infty$ if $v_A=0$.} $\epsilon<|v_A^{-1}|$ for some $\delta \in [0,1)$, $C>0$ not depending on $h$, there exists a neighborhood $U(h) \subset \Delta_n \times \Delta_m$ of $(x^*,y^*)$ such that for all $(x^0,y^0) \in U(h)$, the CMWU iterates (1) converge to the NE, namely $\lim_{t \to \infty} (\varphi_1(x^t,y^t),\varphi_2(x^t,y^t))=(x^*,y^*)$.
\end{theorem}

\begin{proof} The proof aims at applying the well-know proposition \ref{ostro}, as in the OMWU case \cite{omwu}. There are 2 key differences in the proof of CMWU vs.\ OMWU: (i) since the update rule is different, we need a different proof to ensure that any NE is a fixed point of the update rule: this is proved in proposition~1. (ii) regarding the eigenvalue analysis, OMWU attacks the antidiagonal part of the Jacobian. As stated in \cite{omwu2}, which proves local convergence of OMWU in the convex-concave zero-sum setting: \textit{"The key result of spectral analysis in \cite{omwu} is in Lemma B.6 which states that an skew-symmetric matrix (i.e. antisymmetric $A^T=-A$) has imaginary eigenvalues.} In our CMWU case, we will attack the diagonal part of the Jacobian, which does not exploit skew-symmetry and thus necessitates a different proof technique. Specifically, we will see the problem as an eigenvalue perturbation problem, namely the perturbation of the zero eigenvalue of the "all ones" matrix $J_n$. The remainder of the proof aims at computing the Jacobian and analyzing its eigenvalues.

\textbf{\col{Step 1.}} By proposition 1, the NE is a fixed point of the update rule, it thus remains to show, according to proposition \ref{ostro}, that the eigenvalues of the game Jacobian at equilibrium $H=H(x^*,y^*)$ have modulus strictly less then one, where $H$ is defined as:
$$
H:= 
\begin{pmatrix}
\frac{\partial^2 \varphi_1}{\partial x^2} & \frac{\partial^2 \varphi_1}{\partial x \partial y} \\
\frac{\partial^2 \varphi_2}{\partial x \partial y} & \frac{\partial^2 \varphi_2}{\partial y^2}
\end{pmatrix}
$$

Let us first compute $H$ directly. We get, denoting $\delta_{ij}$ the Kroenecker delta ($\delta_{ij}=1$ if $i=j$, $\delta_{ij}=0$ otherwise):
\begin{align*}
\frac{\partial [\varphi_1]_i}{\partial x_j}(x^*,y^*)=\delta_{ij} \gamma_i+x^*_i \gamma_i \left[ -h \epsilon H_{y^*}[i,j] -\gamma_j+h\epsilon \sum_{k=1}^n \gamma_k x^*_k H_{y^*}[k,j]\right], \hspace{2mm} i,j \in [1,n]
\end{align*}
\begin{align*}
\frac{\partial [\varphi_1]_i}{\partial y_j}(x^*,y^*)=&x^*_i \gamma_i \left[ h A_{ij}-h\epsilon A_{ij}[A^Tx^*]_j\right.\\&\left.+h\sum_{k=1}^n x^*_k \gamma_k (-A_{kj}+\epsilon A_{kj}[A^Tx^*]_j)\right], \hspace{2mm} i \in [1,n], j \in [1,m]    
\end{align*}
\begin{align*}
\gamma_i = \frac{\exp(h[Ay^*]_i - h\epsilon [H_{y^*}x^*]_i)}{\sum_{k=1}^n x^*_k\exp(h[Ay^*]_k - h\epsilon [H_{y^*} x^*]_k)}, \hspace{2mm} i \in [1,n]
\end{align*}

$$
\frac{\partial [\varphi_2]_i}{\partial y_j}(x^*,y^*)=\delta_{ij} \eta_i+y^*_i \eta_i \left[ -h \epsilon H_{x^*}[i,j] -\eta_j+h\epsilon \sum_{k=1}^m \eta_k y^*_k H_{x^*}[k,j]\right], \hspace{2mm} i,j \in [1,m]
$$
\begin{align*}
\frac{\partial [\varphi_2]_i}{\partial x_j}(x^*,y^*)=&y^*_i \eta_i \left[ -h A_{ji}-h\epsilon A_{ji}[Ay^*]_j\right.\\&\left.+h\sum_{k=1}^n y^*_k \eta_k (A_{jk}+\epsilon A_{jk}[Ay^*]_j)\right], \hspace{2mm} i \in [1,m], j \in [1,n]
\end{align*}
$$
\eta_i = \frac{\exp(-h[A^Tx^*]_i - h\epsilon [H_{x^*}y^*]_i)}{\sum_{k=1}^m y^*_k\exp(-h[A^Tx^*]_k - h\epsilon [H_{x^*} y^*]_k)}, \hspace{2mm} i \in [1,m]
$$

\textbf{\col{Step 2.}} Eigenspaces are invariant with respect to Jacobian row and column permutations, so we are free to rearrange $H$ as desired. As in \cite{omwu}, we first group together equilibrium coordinates outside of the support, i.e. the indexes $i$ and $j$ such that $x_i^*=0$ or $y_j^*=0$. We denote $\supp(z)$ the support of a vector $z$, namely the indexes corresponding to its nonzero coordinates. Take $i \notin \supp(x^*)$, $j \notin \supp(y^*)$. In this case, using the partial derivatives computed in step 1, all partial derivatives are zero except the diagonal terms:
$$
\frac{\partial [\varphi_1]_i}{\partial x_i}(x^*,y^*)=\gamma_i \hspace{10mm}
\frac{\partial [\varphi_2]_j}{\partial y_j}(x^*,y^*)=\eta_j
$$
We use the NE property that $x_k^*\neq 0 \Rightarrow [Ay^*]_k=v_A$, and $y_k^*\neq 0 \Rightarrow [A^Tx^*]_k=v_A$ \cite{omwu}. This yields, since the sum in the denominator below is equal to the same sum restricted to $x_k^* \neq 0$ (as we multiply by $x_k^*$): 
$$
\gamma_i = \frac{\exp(h[Ay^*]_i - h\epsilon [H_{y^*}x^*]_i)}{\sum_{k=1}^n x^*_k\exp(hv_A - h\epsilon v_A^2)}=\exp(h[Ay^*]_i-hv_A+h\epsilon v_A^2 - h\epsilon [H_{y^*}x^*]_i)
$$
where we have used that $[H_{y^*} x^*]_k=v_A^2$ when $x^*_k \neq 0$, cf. proof of proposition 1. Proof of proposition 1 also yields $[H_{y^*}x^*]_i=[A y^*]_i v_A$ (remember $x_i^* = 0$). This yields:
$$
\gamma_i =\exp(h[Ay^*]_i-hv_A+h\epsilon v_A^2 - v_A h\epsilon [A y^*]_i)=\exp(-h(v_A-[Ay^*]_i)(1-\epsilon v_A))
$$
and therefore $|\gamma_i|<1$ when $(v_A-[Ay^*]_i)(1-\epsilon v_A)>0$. Since the NE $(x^*,y^*)$ is quasi-strict by assumption, $v_A-[Ay^*]_i>0$, hence the condition is true when $1-\epsilon v_A>0$, i.e. $\epsilon v_A<1$. Similarly, we get:
$$
\eta_j=\exp(-h([A^Tx^*]_j-v_A)(1+\epsilon v_A))
$$
By quasi-strict property of the NE, we have $[A^Tx^*]_j-v_A>0$, and so $|\eta_j|<1$ when $1+\epsilon v_A>0$, i.e. $\epsilon v_A>-1$. Both requirements $\epsilon v_A<1$, $\epsilon v_A>-1$ yield $|\epsilon v_A|<1$, which is true by assumption since we have assumed $0<\epsilon <|v_A^{-1}|$ (with the convention $|v_A^{-1}|:=+\infty$ if $v_A=0$).

\textbf{\col{Step 3.}} We now consider the remaining coordinates, i.e. indexes $i$ and $j$ such that both $x_i^*\neq 0$ and $y_j^*\neq 0$. The partial derivatives of step 1 can be simplified: we use again the NE property that $x_i^*\neq 0 \Rightarrow [Ay^*]_i=v_A$, and $y_j^*\neq 0 \Rightarrow [A^Tx^*]_j=v_A$, which yield that all $\gamma$ and $\eta$ terms in these equations are equal to 1. Specifically for $i,j \in \supp(x^*)$:
$$
\frac{\partial [\varphi_1]_i}{\partial x_j}(x^*,y^*)=\delta_{ij} +x^*_i \left[ -h \epsilon H_{y^*}[i,j] -1+h\epsilon \sum_{k=1}^n x^*_k H_{y^*}[k,j]\right] \hspace{2mm} 
$$
$$
=\delta_{ij} +x^*_i \left[ -h \epsilon H_{y^*}[i,j] -1+h\epsilon v_A^2\right]
$$
for $i \in \supp(x^*)$, $j \in \supp(y^*)$:
\begin{align*}
\frac{\partial [\varphi_1]_i}{\partial y_j}(x^*,y^*)&=x^*_i \left[ h A_{ij}-h\epsilon A_{ij}v_A+h\sum_{k=1}^n x^*_k (-A_{kj}+\epsilon A_{kj}v_A)\right]\\
&=x^*_i \left[ h A_{ij}-h\epsilon A_{ij}v_A-hv_A+h \epsilon v_A^2\right] 
\end{align*}
for $i,j \in \supp(y^*)$
$$
\frac{\partial [\varphi_2]_i}{\partial y_j}(x^*,y^*)=\delta_{ij} +y^*_i \left[ -h \epsilon H_{x^*}[i,j] -1+h\epsilon \sum_{k=1}^m y^*_k H_{x^*}[k,j]\right]
$$
$$
=\delta_{ij} +y^*_i \left[ -h \epsilon H_{x^*}[i,j] -1+h\epsilon v_A^2\right]
$$
for $i \in \supp(y^*)$, $j \in \supp(x^*)$:
\begin{align*}
\frac{\partial [\varphi_2]_i}{\partial x_j}(x^*,y^*)=&y^*_i \left[ -h A_{ji}-h\epsilon A_{ji}v_A+h\sum_{k=1}^n y^*_k (A_{jk}+\epsilon A_{jk}v_A)\right]\\
&=y^*_i \left[ -h A_{ji}-h\epsilon A_{ji}v_A+hv_A+h \epsilon v_A^2\right]
\end{align*}
The above expressions can be written in compact matrix forms. In the following, in order to avoid notational burden we will use the same notations for the original matrices and their support restricted counterparts, but in the expressions below, the upper-left, upper-right, lower-left and lower-right matrix blocks have sizes respectively: $|\supp(x^*)|^2$, $|\supp(x^*)| \times |\supp(y^*)|$, $|\supp(y^*)| \times |\supp(x^*)|$, $|\supp(y^*)|^2$. Denoting $J$ as the all ones matrix, we get:
$$
H= I -  D \left[D_J
+  h\epsilon S
- h \epsilon v_A^2 \widetilde{J}
 + hB
\right], \hspace{3mm}
\widetilde{J}:=\begin{pmatrix}
J & J \\
J & J
\end{pmatrix}, \hspace{3mm}
D:=
\begin{pmatrix}
D_{x^*} & 0 \\
0 & D_{y^*}
\end{pmatrix}
$$
$$
S := \begin{pmatrix}
H_{y^*} & v_A A_* \\
v_A A_*^T  & H_{x^*}
\end{pmatrix}, \hspace{3mm}
D_J:=\begin{pmatrix}
J & 0 \\
0 & J
\end{pmatrix}, \hspace{3mm}
B:=\begin{pmatrix}
0 & -A_*+v_AJ \\
A_*^T-v_AJ & 0
\end{pmatrix}
$$

\textbf{\col{Step 4.}} Remember that $D$ has strictly positive diagonal elements due to the support restriction: $D_{x^*}$, $D_{y^*}$ are of sizes $|\supp(x^*)|^2$, $|\supp(y^*)|^2$. Since the sum of $x^*$ and $y^*$ are both 1, eigenvalues of $D D_J$ are $1$ (multiplicity 2) and $0$ (multiplicity $|\supp(x^*)| + |\supp(y^*)|-2$). When $h=0$, the zero eigenvalue of $D D_J$ is causing $H$ to have eigenvalue $1$, so for small $h$, our task is to show that the perturbed eigenvalue $\lambda_h$ away from zero has strictly positive real part, which will make the modulus of the corresponding eigenvalue of $H$ strictly less than one (see corollary 9 in \cite{dmg}). Observe that the zero eigenvalue of $D D_J$ has same eigenvectors as the zero eigenvalue of $D_J$ since $D$ is invertible. $D_J$ is symmetric and bloc diagonal with $J$ on the diagonal, so these eigenvectors are orthogonal and of the form $z:=(z_1,z_2)$ of size $|\supp(x^*)| \times |\supp(y^*)|$ with one of $z_1$, $z_2$ an eigenvector associated to the zero eigenvalue of $J$, and the other 0. Note that the eigenspace $\mathcal{Z}_n$ associated to the zero eigenvalue of $J_n$ is of dimension $n-1$ and consists of vectors whose sum is zero.

\textbf{\col{Step 5.}} Let $\lambda_h$ a perturbation of the 0 eigenvalue of $D_J$ for small $h$, associated to the perturbation $D_J+h\epsilon M+hB$, with $M:=S-v_A^2 \widetilde{J}$. We use lemma \ref{lemeig} (since $D_J$ is symmetric) to get that $\lambda_h = h\epsilon z^T M z + h z^T B z + o(h\epsilon) + o(h)$. By assumption, $\epsilon=C h^{-\delta}$ with $\delta \in [0,1)$, so the term $o(h)$  is $o(h\epsilon)$ and $h\epsilon=o(1)$, hence $\lambda_h = h\epsilon z^T M z + h z^T B z + o(h\epsilon)$. Using the fact that one of $z_1$, $z_2$ is an eigenvector of $J$, and the other 0, we get $z^T \widetilde{J} z=0$. Since $B$ has zero diagonal blocs, its antidiagonal blocs will be multiplied by $z_1$ on one side and $z_2$ on the other side, which yield $z^T B z=0$ since either $z_1=0$ or $z_2=0$. Finally we get $z^T S z = z_1^T H_{y^*} z_1 + z_2^T H_{x^*} z_2$, since its antidiagonal blocs will be multiplied by 0 for the same reason as for $B$.

\textbf{\col{Step 6.}} We have $z_1^T H_{y^*} z_1 = z_1^T A_* D_{y^*}^{\frac{1}{2}} D_{y^*}^{\frac{1}{2}} A_*^T z_1 = ||D_{y^*}^{\frac{1}{2}} A_*^T z_1||^2$. Notice that $||D_{y^*}^{\frac{1}{2}} A_*^T z_1||=0$ if and only if $A_*^T z_1=0$ as $D_{y^*}$ is of size $|\supp(y^*)|^2$ has been restricted to rows and columns in the support of $y^*$. Similarly $||D_{x^*}^{\frac{1}{2}} A_* z_2||=0$ if and only if $A_* z_2=0$. By assumption, $A_*^T$ is weakly $\mathcal{Z}_{|\supp(x^*)|}-$invertible, and $A_*$ is weakly $\mathcal{Z}_{|\supp(y^*)|}-$invertible, hence we always have either $A_* z_2 \neq 0$ or $A_*^Tz_1\neq 0$ since one of $z_1$, $z_2$ is zero and the other one is in $\mathcal{Z}$ and non zero as an eigenvector associated to the zero eigenvalue of $J$. So overall $\lambda_h=h\epsilon \lambda + o(h\epsilon)$ for some $\lambda>0$. Scaling on the left the matrix $D_J$ by $D$ doesn't change the sign of $\lambda$ since we are working on the support of $x^*$, $y^*$. This concludes the proof that for $h$ sufficiently small, the perturbed 0 eigenvalue of $D D_J$ will have strictly positive real part. 

\end{proof}

\begin{lemma}
\label{pla}
If $p_1$, $p_2$ are commutative projectors on $\mathcal{G}$, then $p_1 p_2$ is a projector and:
$$
\mathcal{G}=(K_{p_1} \cap K_{p_2}) \oplus (K_{p_1} \cap R_{p_2}) \oplus (R_{p_1} \cap K_{p_2}) \oplus (R_{p_1} \cap R_{p_2})\ .
$$
\end{lemma}
\begin{proof}
The result is well-known in linear algebra. We recap the proof for convenience. \textbf{Step 1.} Let $u:\mathcal{G} \to \mathcal{G}$ and $q$ a projector. Then $u$ and $q$ commute if and only if $K_q$ and $R_q$ are stable by $u$, i.e. $uK_q \subseteq K_q$, $uR_q \subseteq R_q$. Indeed, assume $u$ and $q$ commute and $x \in K_q$, $y \in R_q$. Then, $qux=uqx=0$, hence $ux \in K_q$. By definition 5, $K_q=R_{\id-q}$, but $\id-q$ is a projector since it is linear and $(\id-q)(\id-q)=\id-2q+q^2=\id-q$. So, $R_q=K_{\id-q}$, and $(\id-q)uy=u(\id-q)y=0$ since $y\in R_q=K_{\id-q}$. Now, assume the reverse. Let $z \in \mathcal{G}$. We want to show $uqz=quz$. By definition 5, $z=x+y$ with $x \in K_q$, $y \in R_q$, so $uqz=uqy$ and $quz=qux+quy$. By assumption $ux \in K_q$, so $qux=0$. So it remains to show $uqy=quy$. $y\in R_q$, so $y=qx_0$ for some $x_0 \in \mathcal{G}$, so $uqy=uqx_0$. Since $uy \in R_q$ by assumption, $quy=uy=uqx_0$, which shows $uqy=quy$.

\textbf{Step 2.} Let $\mathcal{K}$ a subspace of $\mathcal{G}$ and $q$ a projector. We show that $\mathcal{K}$ is stable by $q$ if and only if $\mathcal{K}=(\mathcal{K} \cap K_q) \oplus (\mathcal{K} \cap R_q)$. Indeed, assume $\mathcal{K}=(\mathcal{K} \cap K_q) \oplus (\mathcal{K} \cap R_q)$ and $z \in \mathcal{K}$. We want to show $qz \in \mathcal{K}$. By assumption $z=x+y$, with $x \in \mathcal{K} \cap K_q$, $y \in \mathcal{K} \cap R_q$. So $qz=qx+qy=qy=y \in \mathcal{K}$. Now assume the reverse, i.e. $\mathcal{K}$ is stable by $q$, and let $H:=(\mathcal{K} \cap K_q) + (\mathcal{K} \cap R_q)$. Since $q$ projector, by definition 5 and property of the direct sum we have $K_q \cap R_q=\{0\}$, so the sum is a direct sum: $H=(\mathcal{K} \cap K_q) \oplus (\mathcal{K} \cap R_q)$. The inclusion $H \subseteq \mathcal{K}$ is trivial, so it remains to show $\mathcal{K} \subseteq H$. Let $z \in \mathcal{K}$. Then $z=x+y$ with $x \in K_q$, $y \in R_q$, so $qz=qy=y$. Since $qz \in \mathcal{K}$ by assumption, then $y \in \mathcal{K}$, so $y \in \mathcal{K} \cap R_q$. But $x=z-y$, so also $x \in \mathcal{K}$, and so $x \in \mathcal{K} \cap K_q$, which proves that $z \in H$.

\textbf{Step 3.} We have $\mathcal{G}=K_{p_1} \oplus R_{p_1}$ by definition 5. By step 1, $K_{p_1}$ and $R_{p_1}$ are stable by $p_2$. By step 2, $K_{p_1}=(K_{p_1} \cap K_{p_2}) \oplus (K_{p_1} \cap R_{p_2})$ and $R_{p_1}=(R_{p_1} \cap K_{p_2}) \oplus (R_{p_1} \cap R_{p_2})$, which completes the proof.
\end{proof}

\begin{theorem}
\label{decompa}
\col{\textbf{(Game Decomposition)}}
Let $(\rho_i)_{i \in [1,n]}$ a family of $n$ commutative projectors on $\mathcal{G}$. Then we have the canonical direct sum decomposition:
\begin{align*}
    \mathcal{G}= \bigoplus_{\mathcal{C}_i \in \{K_{\rho_i},R_{\rho_i}\}} \bigcap_{i=1}^n \mathcal{C}_i\ .
\end{align*}
In particular, in the case $n=3$ we have:
\begin{align*}
\begin{split}
\mathcal{G}= (K_{\rho_1} \cap K_{\rho_2} \cap K_{\rho_3}) & \oplus  (K_{\rho_1} \cap K_{\rho_2} \cap R_{\rho_3}) \oplus  (K_{\rho_1} \cap R_{\rho_2} \cap K_{\rho_3}) \oplus  (K_{\rho_1} \cap R_{\rho_2} \cap R_{\rho_3})\\ 
\oplus (R_{\rho_1} \cap K_{\rho_2} \cap K_{\rho_3}) & \oplus  (R_{\rho_1} \cap K_{\rho_2} \cap R_{\rho_3}) \oplus  (R_{\rho_1} \cap R_{\rho_2} \cap K_{\rho_3}) \oplus  (R_{\rho_1} \cap R_{\rho_2} \cap R_{\rho_3})\ .
\end{split}
\end{align*}
\end{theorem}
\begin{proof}
Let us prove the result by induction. The result is true if $n=1$ by definition 5: if $p_1$ is any projector, then $\mathcal{G}= K_{p_1} \oplus R_{p_1}$. The result is also true for $n=2$ by lemma \ref{pla}. Now, assume the result is true at induction stage $n$, and consider a family of $n+1$ commutative projectors $(\rho_i)_{i \in [1,n+1]}$. By the induction hypothesis applied to the family $(\rho_i)_{i \in [1,n]}$, we have:
\begin{align*}
    \mathcal{G}= \bigoplus_{\mathcal{C}_i \in \{K_{\rho_i},R_{\rho_i}\}} \bigcap_{i=1}^n \mathcal{C}_i\
\end{align*}

The goal is now to split each one of the subspaces $\mathcal{C}_i$ into 2 as follows: $\mathcal{C}_i=(\mathcal{C}_i \cap K_{\rho_{n+1}}) \oplus (\mathcal{C}_i \cap R_{\rho_{n+1}})$. If the latter is true, then the induction is true at stage $n+1$. By the step 2 of the proof of lemma \ref{pla}, we have $\mathcal{C}_i=(\mathcal{C}_i \cap K_{\rho_{n+1}}) \oplus (\mathcal{C}_i \cap R_{\rho_{n+1}})$ provided $\mathcal{C}_i$ is stable by $\rho_{n+1}$, i.e. $\rho_{n+1}\mathcal{C}_i \subseteq \mathcal{C}_i$. We now conclude the proof using a similar argument as step 1 of the proof of lemma \ref{pla}, since by assumption $\rho_{n+1}$ commutes with all $(\rho_{i})_{i \in [1,n]}$. Indeed, by definition of $\mathcal{C}_i$, it is an intersection of the following form, for some integers $n_1$, $\alpha_{k}$, $\beta_{k}$:
$$
\mathcal{C}_i = \cap_{k=1}^{n_1} K_{\rho_{\alpha_{k}}} \cap \cap_{k=1}^{n-n_1} R_{\rho_{\beta{k}}}
$$
Take $z \in \mathcal{C}_i$. The goal is to show that $\rho_{n+1} z \in \mathcal{C}_i$. Since $\rho_{n+1}$ commutes with every $\rho_{\alpha_{k}}$, we have $\rho_{\alpha_{k}} \rho_{n+1} z = \rho_{n+1} \rho_{\alpha_{k}}z=0$ since $z \in K_{\rho_{\alpha_{k}}}$, which shows that $\rho_{n+1} z \in \cap_{k=1}^{n_1} K_{\rho_{\alpha_{k}}}$. Similarly, since $R_{\rho_{\beta{k}}} = K_{\id -\rho_{\beta{k}}}$, we have $(\id-\rho_{\beta_{k}}) \rho_{n+1} z = \rho_{n+1} (\id-\rho_{\beta_{k}})z=0$ since $z \in R_{\rho_{\beta{k}}}$, which shows that $\rho_{n+1} z \in \cap_{k=1}^{n-n_1} R_{\rho_{\beta_{k}}}$, which concludes the proof.
\end{proof}

\begin{proposition}
\col{\textbf{(generalization of transitivity/cyclicity and unification with triviality)}} If a 2P game $f$ is zero-sum symmetric transitive (resp. cyclic) in the sense of \cite{2pszs}, then $f \in R_{\rho_{T}}$ (resp. $f \in K_{\rho_{T}}=R_{\id-\rho_T}$). Further, a 2P game is trivial in the sense of \cite{chaos} if and only if it is transitive, i.e. the class of trivial games coincides with $R_{\rho_{T}}$.
\end{proposition}
\begin{proof}
Remember the following definitions. $\mathcal{N}_2f(x,y):= \int f(x,z)d\mu_2(z)$, $\mathcal{N}_1f(x,y):= \int f(z,y)d\mu_1(z)$, and:
\begin{align*}
\bm{\col{\rho_{T}}}: (f_1,f_2) \to (\widehat{\mathcal{N}}f_1, \widehat{\mathcal{N}}f_2), \hspace{2mm} \widehat{\mathcal{N}}:= \mathcal{N}_1+\mathcal{N}_2-\mathcal{N}_1\mathcal{N}_2\ .
\end{align*}
First, if a 2P game $f=(f_1,f_2)$ is zero-sum symmetric transitive in the sense of \cite{2pszs}, then $\mathcal{X}_1=\mathcal{X}_2$, $\mu_1=\mu_2=:\mu$, $f_2=-f_1$ and $f_1(x,y)=\phi(x)-\phi(y)$ for some function $\phi$. We show that $\rho_T f=f$, namely that $f \in R_{\rho_T}$. For this, it suffices to show $\widehat{\mathcal{N}}f_1=f_1$. Indeed, $f_2=-f_1$ so $\widehat{\mathcal{N}}f_1=f_1$ yields $\widehat{\mathcal{N}}f_2=f_2$ by linearity. We have: 
\begin{align*}
\widehat{\mathcal{N}} f_1(x,y)&= \int f_1(x,z)d\mu(z) + \int f_1(z,y)d\mu(z)-\int \int f_1(z_1,z_2)d\mu(z_1)d\mu(z_2)\\
&=\phi(x)-\int \phi(z)d\mu(z) + \int \phi(z)d\mu(z)-\phi(y) - \left[\int \phi(z)d\mu(z)-\int \phi(z)d\mu(z)\right]\\
&=\phi(x)-\phi(y)=f_1(x,y)
\end{align*}
If the game is zero-sum symmetric cyclic in the sense of \cite{2pszs}, then $f_2=-f_1$, $f_1=-f_1^T$, and $\int f_1(x,z)d\mu(z)=0$ $\forall x \in \mathcal{X}_1$, so we have:
\begin{align*}
\widehat{\mathcal{N}} f_1(x,y)&= \int f_1(x,z)d\mu(z) + \int f_1(z,y)d\mu(z)-\int \int f_1(z_1,z_2)d\mu(z_1)d\mu(z_2)\\
&=0-\int f_1(y,z)d\mu(z)+\int \left(\int f_1(z_2,z_1)d\mu(z_1)\right)d\mu(z_2)\\
&=0-0+\int 0 \hspace{1mm} d\mu(z_2)=0
\end{align*}
hence indeed $f \in K_{\rho_T}$ ($f_2=-f_1$ so $\widehat{\mathcal{N}}f_2=-\widehat{\mathcal{N}}f_1=0$ by linearity).

Now, we show that $(f_1,f_2)$ is trivial in the sense of \cite{chaos} if and only if it is transitive. A trivial game is of the form $f_1(x,y)=\alpha_1(x)+\beta_1(y)$, $f_2(x,y)=\alpha_2(x)+\beta_2(y)$ for some functions $\alpha_i$, $\beta_i$. We have:
$$
\widehat{\mathcal{N}} f_k(x,y)= \int f_k(x,z)d\mu_2(z) + \int f_k(z,y)d\mu_1(z)-\int \int f_k(z_1,z_2)d\mu_1(z_1)d\mu_2(z_2)
$$
Therefore clearly, if a game belongs to the range $R_{\rho_T}$ of $\rho_T$, it is trivial since $\widehat{\mathcal{N}} f_k$ is of the form $\alpha_k(x)+\beta_k(y)$ for $\alpha_k(x)=\int f_k(x,z)d\mu_2(z)$, $\beta_k(y)= \int f_k(z,y)d\mu_1(z)-\int \int f_k(z_1,z_2)d\mu_1(z_1)d\mu_2(z_2)$. It remains to show the opposite, i.e. that if a game is trivial, it belongs to $R_{\rho_T}$, i.e. that $\widehat{\mathcal{N}} f_1=f_1$ and $\widehat{\mathcal{N}} f_2=f_2$. We have:
\begin{align*}
&\widehat{\mathcal{N}} f_k(x,y)= \int f_k(x,z)d\mu_2(z) + \int f_k(z,y)d\mu_1(z)-\int \int f_k(z_1,z_2)d\mu_1(z_1)d\mu_2(z_2)\\
&=\alpha_k(x) + \int \beta_k(z)d\mu_2(z) + \beta_k(y) + \int \alpha_k(z)d\mu_1(z)- \left[\int \alpha_k(z)d\mu_1(z) + \int \beta_k(z)d\mu_2(z)\right] \\&= \alpha_k(x)+\beta_k(y)=f_k(x,y)
\end{align*}
This concludes the proof.
\end{proof}

\begin{corollary}
\col{\textbf{(Canonical decomposition of 2P games)}}
$\rho_{Z}$, $\rho_{T}$ are commutative projectors. Consequently, by theorem \ref{decompa}, every 2P game can be decomposed into the direct sum of 4 components: \col{\textbf{(ZT)}} zero-sum-transitive, \col{\textbf{(ZCy)}} zero-sum-cyclic, \col{\textbf{(CT)}} cooperative-transitive, \col{\textbf{(CCy)}} cooperative-cyclic. Further, if $\mathcal{X}_1=\mathcal{X}_2$, then $\rho_{S}$ is a well-defined projector and commutes with $\rho_{Z}$, $\rho_{T}$. Consequently, any such game can be decomposed into the direct sum of 8 components: \col{\textbf{(ZST)}} zero-sum-symmetric-transitive, \col{\textbf{(ZSCy)}} zero-sum-symmetric-cyclic, \col{\textbf{(ZAT)}} zero-sum-antisymmetric-transitive, \col{\textbf{(ZACy)}} zero-sum-antisymmetric-cyclic, 
\col{\textbf{(CST)}} cooperative-symmetric-transitive, \col{\textbf{(CSCy)}} cooperative-symmetric-cyclic, \col{\textbf{(CAT)}} cooperative-antisymmetric-transitive, \col{\textbf{(CACy)}} cooperative-antisymmetric-cyclic.
\end{corollary}
\begin{proof}
Remember that we have the following definitions for the operators corresponding to zero-sum, symmetric and transitive games:
\begin{align*}
\begin{split}
& \bm{\col{\rho_{Z}}}: (f_1,f_2) \to \frac{1}{2}(f_1-f_2, f_2-f_1), \hspace{2mm} 
\bm{\col{\rho_{S}}}: (f_1,f_2) \to \frac{1}{2}(f_1+f_2^T, f_2+f_1^T)\\
& \bm{\col{\rho_{T}}}: (f_1,f_2) \to (\widehat{\mathcal{N}}f_1, \widehat{\mathcal{N}}f_2), \hspace{2mm} \widehat{\mathcal{N}}:= \mathcal{N}_1+\mathcal{N}_2-\mathcal{N}_1\mathcal{N}_2\ .
\end{split}
\end{align*}
where $\mathcal{N}_2f(x,y):= \int f(x,z)d\mu_2(z)$, $\mathcal{N}_1f(x,y):= \int f(z,y)d\mu_1(z)$. For convenience we write below the complementary (in the sense of direct sum) operators $\id-\rho_{Z}$, $\id-\rho_{S}$, $\id-\rho_{T}$ corresponding respectively to cooperative, antisymmetric and cyclic games:
\begin{align*}
\begin{split}
& \id-\rho_{Z}: (f_1,f_2) \to \frac{1}{2}(f_1+f_2, f_1+f_2), \hspace{2mm} 
\id-\rho_{S}: (f_1,f_2) \to \frac{1}{2}(f_1-f_2^T, f_2-f_1^T)\\
& \id-\rho_{T}: (f_1,f_2) \to (f_1-\widehat{\mathcal{N}}f_1, f_2-\widehat{\mathcal{N}}f_2)
\end{split}
\end{align*}

We need to check that $\rho_{Z}$, $\rho_{T}$, $\rho_{S}$ are commutative projectors. The fact that $\rho_{Z}$, $\rho_{T}$ are commutative projectors has been observed in \cite{proj}. First let's check that they are projectors. They are indeed linear, so it remains to show that they verify $\rho^2=\rho$. We do so by direct computation:
\begin{align*}
&\rho_Z^2(f_1,f_2)=\frac{1}{2}\rho_Z(f_1-f_2,f_2-f_1)=\frac{1}{4}(2f_1-2f_2,2f_2-2f_1)\\&=\frac{1}{2}(f_1-f_2,f_2-f_1)=\rho_Z(f_1,f_2)\\
&\rho_S^2(f_1,f_2)=\frac{1}{2}\rho_S(f_1+f_2^T,f_2+f_1^T)=\frac{1}{4}(f_1+(f_1^T)^T+2f_2^T,f_2+(f_2^T)^T+2f_1^T)\\
&=\frac{1}{2}(f_1+f_2^T,f_2+f_1^T)=\rho_S(f_1,f_2)\\
\end{align*}
For $\rho_T$, we have $\rho_T^2(f_1,f_2)=(\widehat{\mathcal{N}}^2f_1,\widehat{\mathcal{N}}^2f_2)$. We show that $\widehat{\mathcal{N}}^2f_k=\widehat{\mathcal{N}}f_k$. Denote $\gamma_k:=\int \int f_k(z_1,z_2)d\mu_1(z_1)d\mu_2(z_2)$. We have:
\begin{align*}
& \widehat{\mathcal{N}}^2f_k(x,y) = \int f_k(x,z_2)d\mu_2(z_2) + \gamma_k -\gamma_k
+ \gamma_k + \int f_k(z_1,y)d\mu_1(z_1)-\gamma_k
-(\gamma_k + \gamma_k -\gamma_k)\\
& = \int f_k(x,z_2)d\mu_2(z_2) + \int f_k(z_1,y)d\mu_1(z_1) - \gamma_k=\widehat{\mathcal{N}}f_k(x,y)
\end{align*}

Finally, we check commutativity. We have:
$$
\rho_Z \rho_S(f_1,f_2)=\frac{1}{2}\rho_Z(f_1+f_2^T,f_2+f_1^T)=\frac{1}{4}(f_1+f_2^T-f_2-f_1^T,-f_1-f_2^T+f_2+f_1^T)
$$
$$
\rho_S \rho_Z(f_1,f_2)=\frac{1}{2}\rho_S(f_1-f_2,f_2-f_1)=\frac{1}{4}(f_1+f_2^T-f_2-f_1^T,-f_1-f_2^T+f_2+f_1^T)
$$
Hence indeed $\rho_Z$ and $\rho_S$ commute. Then, using linearity of $\widehat{\mathcal{N}}$:
$$
\rho_T\rho_Z(f_1,f_2)=\frac{1}{2}(\widehat{\mathcal{N}}f_1-\widehat{\mathcal{N}}f_2,\widehat{\mathcal{N}}f_2-\widehat{\mathcal{N}}f_1)
$$
$$\rho_Z\rho_T(f_1,f_2)=\rho_Z(\widehat{\mathcal{N}}f_1,\widehat{\mathcal{N}}f_2)=\frac{1}{2}(\widehat{\mathcal{N}}f_1-\widehat{\mathcal{N}}f_2,\widehat{\mathcal{N}}f_2-\widehat{\mathcal{N}}f_1)
$$
hence indeed $\rho_Z$ and $\rho_T$ commute. Finally:
$$
\rho_T\rho_S(f_1,f_2)=\frac{1}{2}(\widehat{\mathcal{N}}f_1+\widehat{\mathcal{N}}f_2^T,\widehat{\mathcal{N}}f_2+\widehat{\mathcal{N}}f_1^T)$$
$$\rho_S\rho_T(f_1,f_2)=\rho_S(\widehat{\mathcal{N}}f_1,\widehat{\mathcal{N}}f_2)=\frac{1}{2}(\widehat{\mathcal{N}}f_1+(\widehat{\mathcal{N}}f_2)^T,\widehat{\mathcal{N}}f_2+(\widehat{\mathcal{N}}f_1)^T)
$$
So it remains to show $(\widehat{\mathcal{N}}f_k)^T=\widehat{\mathcal{N}}f_k^T$. Remember that $\rho_T$ is well-defined only when $\mathcal{X}_1=\mathcal{X}_2=:\mathcal{X}$, $\mu_1=\mu_2=:\mu$. In that case we have:
$$
\widehat{\mathcal{N}}f_k^T(x,y)=\int f_k^T(x,z)d\mu(z) + \int f_k^T(z,y)d\mu(z)-\int \int f_k^T(z_1,z_2)d\mu(z_1)d\mu(z_2) 
$$
$$
=\int f_k(z,x)d\mu(z) + \int f_k(y,z)d\mu(z)-\int \int f_k(z_2,z_1)d\mu(z_1)d\mu(z_2)
$$
$$= \widehat{\mathcal{N}}f_k(y,x)=(\widehat{\mathcal{N}}f_k)^T(x,y)
$$
This concludes the proof.
\end{proof}

\section{Experiment details and additional findings}
\label{seca3}

\subsection{Zero-sum bimatrix games}

We report in figure \ref{fz1} the average performance and standard deviation of $\beta_\tau$ ($\tau=500$) of all algorithms on 200 randomly sampled bimatrix zero-sum games in dimension $n=m=$ 25, 50, 75, 100, for a wide range of learning rates $h$. We see that CMWU performs well across a wide band of learning rates.

\subsection{Learning Games}\label{sec:suppl:learning_games}

\textbf{Terminology.} We use the terminology of corollary 2.1, namely \col{\textbf{(ZT)}} zero-sum-transitive, \col{\textbf{(ZCy)}} zero-sum-cyclic, \col{\textbf{(CT)}} cooperative-transitive, \col{\textbf{(CCy)}} cooperative-cyclic, \col{\textbf{(ZST)}} zero-sum-symmetric-transitive, \col{\textbf{(ZSCy)}} zero-sum-symmetric-cyclic, \col{\textbf{(ZAT)}} zero-sum-antisymmetric-transitive, \col{\textbf{(ZACy)}} zero-sum-antisymmetric-cyclic, 
\col{\textbf{(CST)}} cooperative-symmetric-transitive, \col{\textbf{(CSCy)}} cooperative-symmetric-cyclic, \col{\textbf{(CAT)}} cooperative-antisymmetric-transitive, \col{\textbf{(CACy)}} cooperative-antisymmetric-cyclic.

\textbf{Computation of game components.} Let $(A,B)$ a bimatrix game. The operators $\rho_Z$, $\rho_S$, $\rho_T$ of section 3 act as defined below in the bimatrix case (note that as mentioned in corollary 2.1, $\rho_S$ is well defined only if $n=m$). As mentioned in the main text, our concepts of cyclic and transitive game (materialized by the operators $\id-\rho_T$ and $\rho_T$) generalize that of \cite{2pszs} which considered the zero-sum symmetric case: in the case where the game $(A,B)$ is zero-sum symmetric, then $B=-A$, $A=-A^T$, and $\rho_T$, $\id-\rho_T$ reduce to \cite{2pszs}, whereas in general they yield different expressions. The zero-sum symmetric transitive and cyclic components of the general game $(A,B)$ are simply given by $\rho_Z \rho_S \rho_T (A,B)$ and $\rho_Z \rho_S (\id-\rho_T) (A,B)$. This way, we can compute all $8$ components of corollary 2.1 simply by applying these operators one after the other, where the order doesn't change as they commute.
$$
\rho_Z(A,B)=\frac{1}{2}(A-B,B-A), \hspace{3mm} \rho_C(A,B):=[\id-\rho_Z](A,B)=\frac{1}{2}(A+B,A+B)
$$
$$
\rho_S(A,B)=\frac{1}{2}(A+B^T,B+A^T), \hspace{3mm} \rho_A(A,B):=[\id-\rho_S](A,B)=\frac{1}{2}(A-B^T,B-A^T)
$$
$$
\rho_T(A,B)=(\widehat{A}^{(1)}+\widehat{A}^{(2)}-\widehat{A},
\widehat{B}^{(1)}+\widehat{B}^{(2)}-\widehat{B}),
$$
$$
\rho_{Cy}(A,B):=[\id-\rho_T](A,B)=(A-\widehat{A}^{(1)}-\widehat{A}^{(2)}+\widehat{A},
B-\widehat{B}^{(1)}-\widehat{B}^{(2)}+\widehat{B})
$$

where $\widehat{A}^{(1)}$, $\widehat{A}^{(2)}$, $\widehat{A}$ are the matrices with entries $\widehat{A}^{(1)}_{ij}:=\frac{1}{m}\sum_{j=1}^m A_{ij}$, $\widehat{A}^{(2)}_{ij}:=\frac{1}{n}\sum_{i=1}^n A_{ij}$, $\widehat{A}_{ij}:=\frac{1}{mn}\sum_{i=1}^n \sum_{j=1}^m A_{ij}$, and $\rho_C$, $\rho_A$, $\rho_{Cy}$ are associated to cooperative, antisymmetric and Cyclic games. 

With these definitions, the 8 components of corollary 2.1 are simply computed by composition of above projectors, namely: \col{\textbf{(ZT)}} $\rho_Z \rho_T(A,B)$, \col{\textbf{(ZCy)}} $\rho_Z (\id-\rho_T)(A,B)$, \col{\textbf{(CT)}} $(\id-\rho_Z) \rho_T(A,B)$, \col{\textbf{(CCy)}} $(\id-\rho_Z)(\id-\rho_T)(A,B)$, \col{\textbf{(ZST)}} $\rho_Z \rho_S \rho_T (A,B)$, \col{\textbf{(ZSCy)}} $\rho_Z \rho_S (\id-\rho_T) (A,B)$, \col{\textbf{(ZAT)}} $\rho_Z (\id-\rho_S) \rho_T (A,B)$, \col{\textbf{(ZACy)}} $\rho_Z (\id-\rho_S) (\id-\rho_T) (A,B)$, 
\col{\textbf{(CST)}} $(\id-\rho_Z) \rho_S \rho_T (A,B)$, \col{\textbf{(CSCy)}} $(\id-\rho_Z) \rho_S (\id-\rho_T) (A,B)$, \col{\textbf{(CAT)}} $(\id-\rho_Z) (\id-\rho_S) \rho_T (A,B)$, \col{\textbf{(CACy)}} $(\id-\rho_Z) (\id-\rho_S) (\id-\rho_T) (A,B)$.

Following definitions 6 and 7, the \col{\textbf{game signature}} is taken to be the vector of size 8 containing the norms $\frac{1}{2}(||A_i||+||B_i||)$ of the components $(A_i,B_i)_{i \in [1,8]}$, where the matrix norm is chosen to be the $L_2$ norm. We further divide these 8 norms by their overall sum so as to interpret them as weights.

\textbf{Training details.} The RL policy was trained in the RLlib framework \cite{rllib}, and run on AWS using a EC2 C5 24xlarge instance with 96 CPUs, using Proximal Policy Optimization \cite{ppo}. We considered square games of size $n=m=10$. We used configuration parameters in line with \cite{ppo}, that is a clip parameter of 0.3, an adaptive KL penalty with a KL target of $0.01$ and a learning rate of $10^{-6}$. Episodes were taken of length $\tau=1000$ time steps, using $B=90$ parallel runs in between policy updates ($12000$ policy updates). As a result, each policy update was performed with a batch size of $1000 \cdot 90$ timesteps, together with 30 iterations of stochastic gradient descent with SGD minibatch size of 8192. We used a fully connected neural net with 2 hidden layers, 256 nodes per layer, and $\tanh$ activation. Since our action space is continuous, the outputs of the neural net are the mean and stDev of a standard normal distribution, which is then used to sample actions probabilistically (the covariance matrix across actions is chosen to be diagonal).

At test time, we use a different budget for each game type since they are more or less complex and hence take more or less time to converge). These are reported in table \ref{tabbu}.

\begin{table}[!ht]
\caption{Budget $\tau$ used at test time.}
\label{tabbu}
\begin{center}
\begin{tabular}{ll}
\hline
\makecell{Game Type} & \makecell{Budget ($10^2$ iterations)} \\ 
\hline
\textcolor{hanpurple}{ZT}&  2\\  
\textcolor{darktangerine}{ZCy}&  40\\ 
 \textcolor{darkpastelgreen}{CT}&  2\\  
\textcolor{deepcerise}{CCy}& 8\\ 
\textcolor{hanpurple}{ZT} + \textcolor{darktangerine}{ZCy}&40\\  
\textcolor{hanpurple}{ZT} + \textcolor{darkpastelgreen}{CT}& 4\\ 
\textcolor{hanpurple}{ZT} + \textcolor{deepcerise}{CCy}&  4\\ 
 \textcolor{darktangerine}{ZCy} + \textcolor{darkpastelgreen}{CT}& 40\\ 
 \textcolor{darktangerine}{ZCy} + \textcolor{deepcerise}{CCy}& 40\\  
\textcolor{darkpastelgreen}{CT} + \textcolor{deepcerise}{CCy}&  4\\ 
\textcolor{hanpurple}{ZT} + \textcolor{darktangerine}{ZCy} + \textcolor{darkpastelgreen}{CT}& 40\\ 
\textcolor{hanpurple}{ZT} + \textcolor{darktangerine}{ZCy} + \textcolor{deepcerise}{CCy}& 40\\ 
\textcolor{hanpurple}{ZT} + \textcolor{darkpastelgreen}{CT} + \textcolor{deepcerise}{CCy}& 4\\ 
\textcolor{darktangerine}{ZCy} + \textcolor{darkpastelgreen}{CT} + \textcolor{deepcerise}{CCy}&  40\\
\textcolor{hanpurple}{ZT} + \textcolor{darktangerine}{ZCy} + \textcolor{darkpastelgreen}{CT} + \textcolor{deepcerise}{CCy}&  40\\
 \hline
\end{tabular}\end{center}
\end{table}

\textbf{Additional results.}  In figures \ref{l22}-\ref{l33} we report the same metrics as in figures \ref{fff2} and \ref{ff2}, but for all other game mixtures considered in table \ref{tab1}. In figures \ref{ff2}, \ref{l22}, \ref{l32}, for each seed, we compute the mean value and correlation between time-series of coefficients, which are then averaged over seeds. 

\subsection{SHAP Analysis}
In addition to the results in Section~\ref{sec:suppl:learning_games} and as presented in figure \ref{fff21}, we also performed a SHAP analysis of the learned policy in the case of Full-RL CMWU. This explainability method --- as first introduced in \cite{shap} --- aids in understanding the marginal contribution of each input feature to the output of a machine learning model. In this case, we were interested in two key questions:
\begin{enumerate}
    \item Which features were most important on a global level?
    \item How do the importances of each feature evolve as a function of the iteration during a given rollout?
\end{enumerate}
The former is a common question to ask, and provides an overall perspective on the behaviour of the learnt policy. The latter is much more refined and, as we shall see (below), reveals that there is non-trivial use of features in the early stages of a trajectory. What's more, these importance processes appear to vary significantly between the 8 different game types.

In all cases presented below, the SHAP values were computed using the ``KernelExplainer'' provided by the \textsf{SHAP} package~\cite{shap}. For each unique game type we sampled 200 trajectories of 2000 steps as generated by the policy to construct a dataset. To simplify the analysis of the SHAP values, we also grouped features into
\begin{enumerate}
    \item (\textbf{BR}) the best response gaps $\delta_1(x^t,y^t)$, $\delta_2(x^t,y^t)$ ;
    \item (\textbf{GRAD}) the current gradients $Ay^t$, $B^Tx^t$;
    \item (\textbf{HESS}) the current Hessians $H_{x^t}=A^TD_{x^t}A$, $H_{y^t}=BD_{y^t}B^T$;
    \item (\textbf{PAY}) the current payoffs $x^{t,T}Ay^t$, $x^{t,T}By^t$;
\end{enumerate} see Section~3.2 for more details. This aggregation was performed by \emph{summing} over the contribution of each component of the feature groups in one of two ways: either sum over the absolute SHAP values as in figure \ref{fff21}, or take the absolute value after summing over the signed SHAP values as in figures \ref{fig:shap:signed_grad}-\ref{fig:shap:signed_hess}. The latter can always be done while retaining consistency due the additivity property of Shapley values. It is less clear what the former encodes, since we may lose correlation interaction effects, but these values are informative nonetheless. As such, we include both sets of results.

\paragraph{Global importance.}
The global importance across feature groups and game types is summarised in Figures~\ref{fff21} and \ref{fig:shap:signed_grad}-\ref{fig:shap:signed_hess}. Observe that while the contributions do differ between the two aggregation approaches, the distribution is broadly the same. We first note that \textsc{ZST} is unique amongst the eight game types in that almost all emphasis is placed on the gradient inputs; this is true for the gradient and Hessian outputs $G_k=h_k$ and $H_k=-h_k \epsilon_k$. In general, this is considered to be the ``easiest'' of the eight core types, so it is perhaps not so surprising. We also find that there is a tangible difference between the importance of inputs between the learnt gradient and Hessian coefficients. For example, more emphasis is put on the payoff terms with respect the Hessian, and, conversely, the best response gaps only appear significantly for the gradient outputs. There also appears to be quite different behaviour between the learnt coefficients for zero-sum versus co-operative games, while in other cases (such as \textsc{CST} and \textsc{CAT} compared to \textsc{CSCy} and \textsc{CACy}) where there appear to be emergent symmetries. It is of course hard to draw strict conclusions from these phenomena, but it is absolutely clear that the network has learnt tangibly different policies as a function of the game signature. 

\paragraph{Iteration-level importance.}
The second type of analysis we performed focused on the evolution of feature importance as a function of intra-trajectory iteration. The results, which are summarised in Figure~\ref{fig:shap:abs_evolution}, suggest a number of intriguing properties of the learnt policy. Firstly, we find that in all cases there exists a transient regime in the first $\sim 100$ iterations. This is particularly noticeable in the \textsc{ZST} case in which the gradient importance grows very sharply from around 70\% to nearly 100\%. This rapid change in the early stages of learning is consistent across most games. However, it is clear that the zero-sum and cyclic variants have a much longer transient period. Indeed, for \textsc{ZSCy} it takes nearly 2000 steps to stabilise to some fixed limit point. This aligns closely with the intuition that zero-sum/cyclic games are generally harder to solve.

\begin{figure}[ht]
\subfigure{\includegraphics[width=1.\columnwidth]{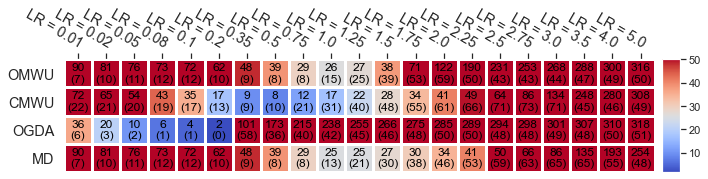}}
\subfigure{\includegraphics[width=1.\columnwidth]{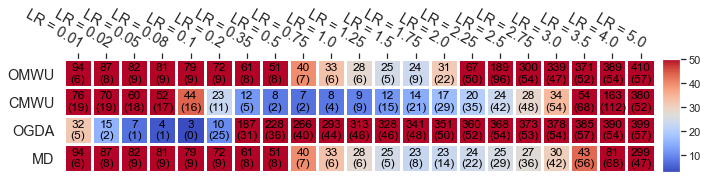}}
\subfigure{\includegraphics[width=1.\columnwidth]{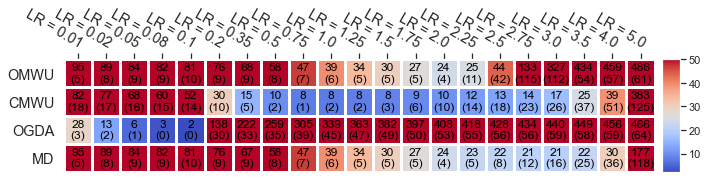}}
\subfigure{\includegraphics[width=1.\columnwidth]{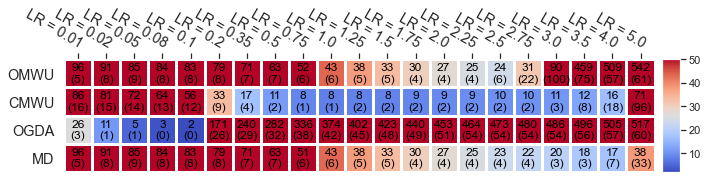}}
\caption{Zero-sum case: average score $\beta_\tau$ across various learning rates $h$ and game sizes $n$. $\tau=500$. $\epsilon=0.25 \cdot h^{-1}n$. StDev in brackets. From \textbf{Top} to \textbf{Bottom}: $n=25,50,75,100$.}
\label{fz1}
\end{figure}

\begin{figure}[ht]
  \centering
  \centerline{\includegraphics[width=1.\columnwidth]{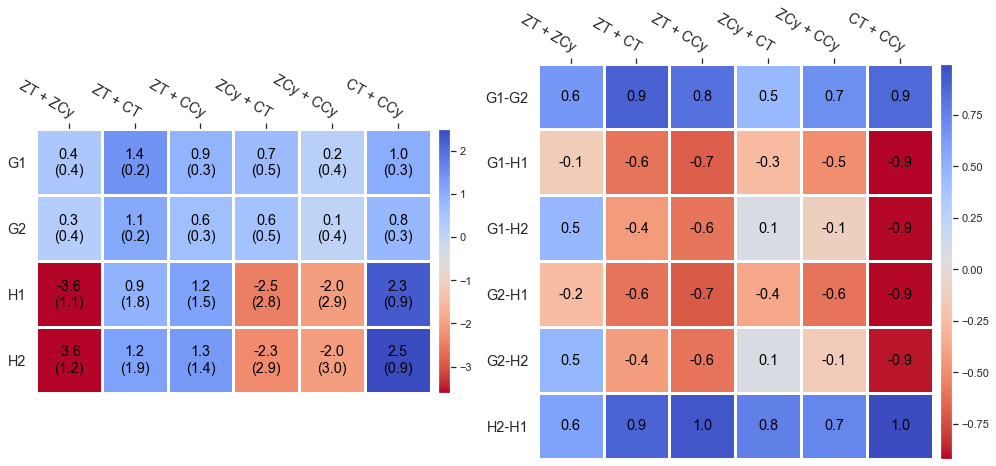}}
  \caption{full-RL CMWU coefficients of gradient $G_k=h_k$ and Hessian $H_k=-h_k \epsilon_k $ of players $k=1,2$ across game types (8 pure game components - mixtures of 2). \textbf{(left)} Mean value. \textbf{(right)} Pairwise correlations. StDev in brackets.}
  \label{l22}
\end{figure}
\begin{figure}[ht]
  \centering
  \centerline{\includegraphics[width=1.\columnwidth]{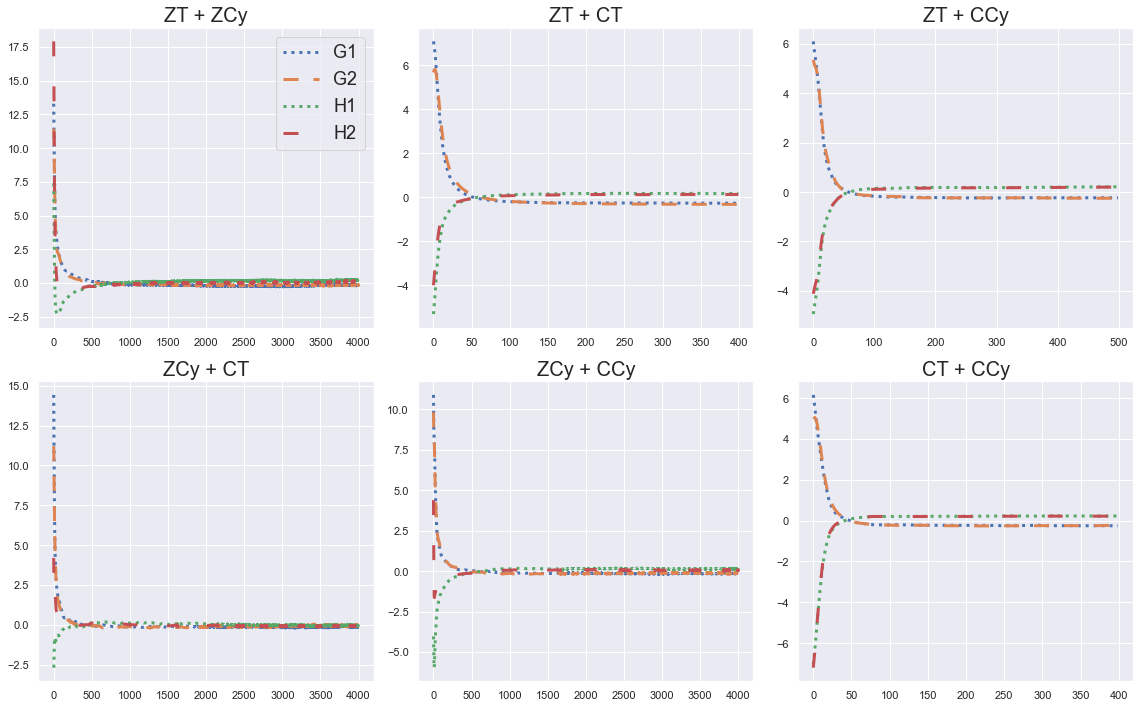}}
  \caption{Average standardized per-episode-trajectory ("shape") of coefficients of gradient $G_k=h_k$ and Hessian $H_k=-h_k \epsilon_k $ of players $k=1,2$ across game types (8 pure game components - mixtures of 2), as a function of time $t$.}
  \label{l23}
\end{figure}

\begin{figure}[ht]
  \centering
  \centerline{\includegraphics[scale=0.4]{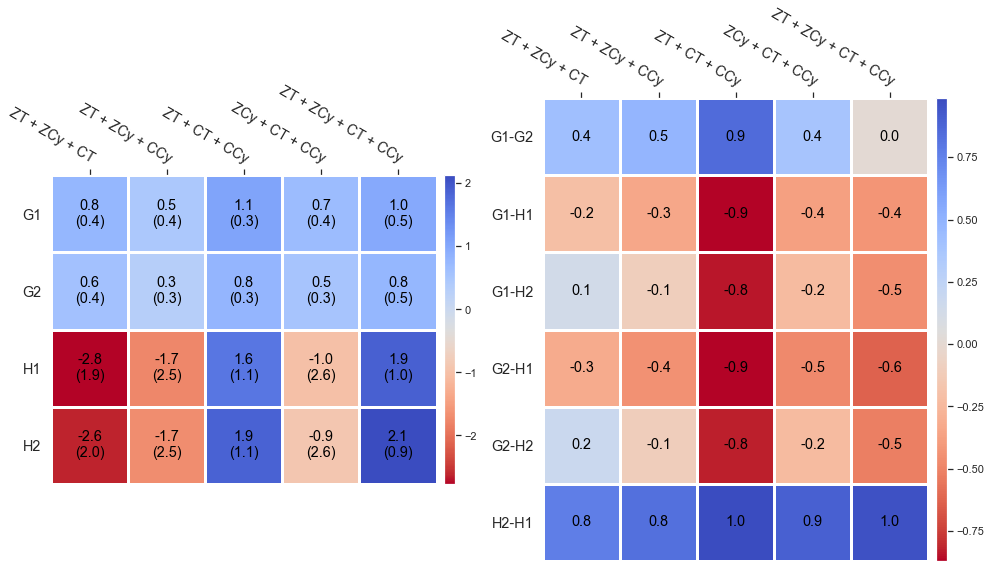}}
  \caption{full-RL CMWU coefficients of gradient $G_k=h_k$ and Hessian $H_k=-h_k \epsilon_k $ of players $k=1,2$ across game types (8 pure game components - mixtures of 3-4). \textbf{(left)} Mean value. \textbf{(right)} Pairwise correlations. StDev in brackets.}
  \label{l32}
\end{figure}
\begin{figure}[ht]
  \centering
  \centerline{\includegraphics[scale=0.4]{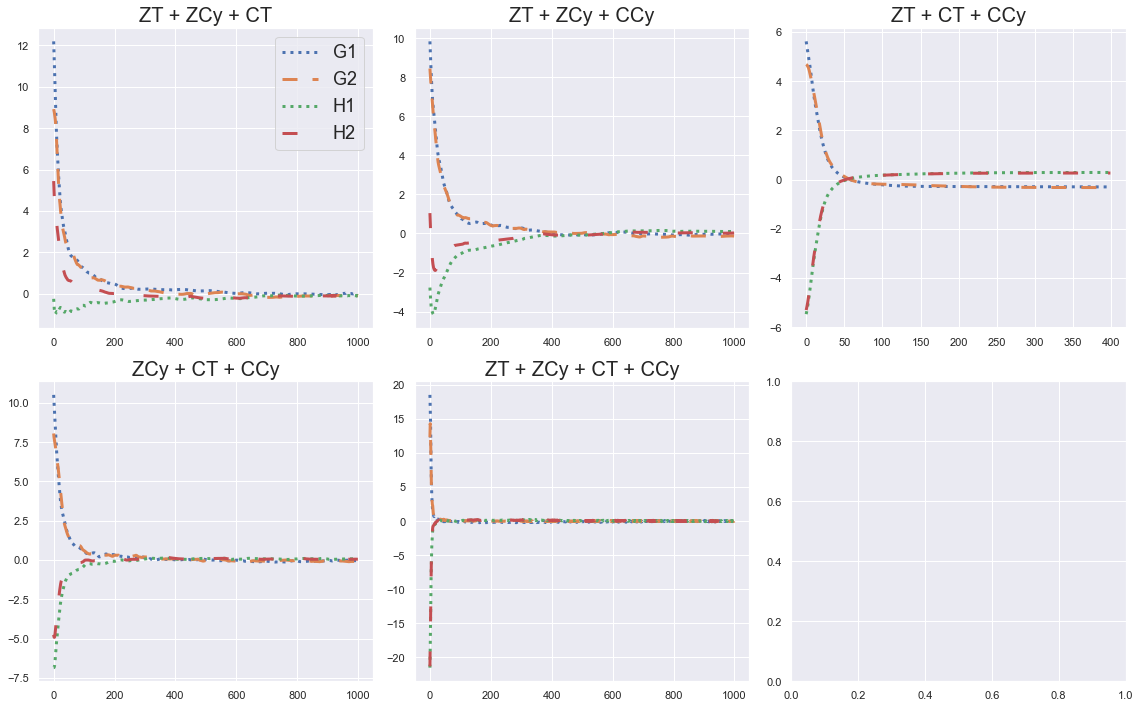}}
  \caption{Average standardized per-episode-trajectory ("shape") of coefficients of gradient $G_k=h_k$ and Hessian $H_k=-h_k \epsilon_k $ of players $k=1,2$ across game types (8 pure game components - mixtures of 3-4), as a function of time $t$.}
  \label{l33}
\end{figure}


\begin{figure}[ht]
  \centering
  \centerline{\includegraphics[scale=0.4]{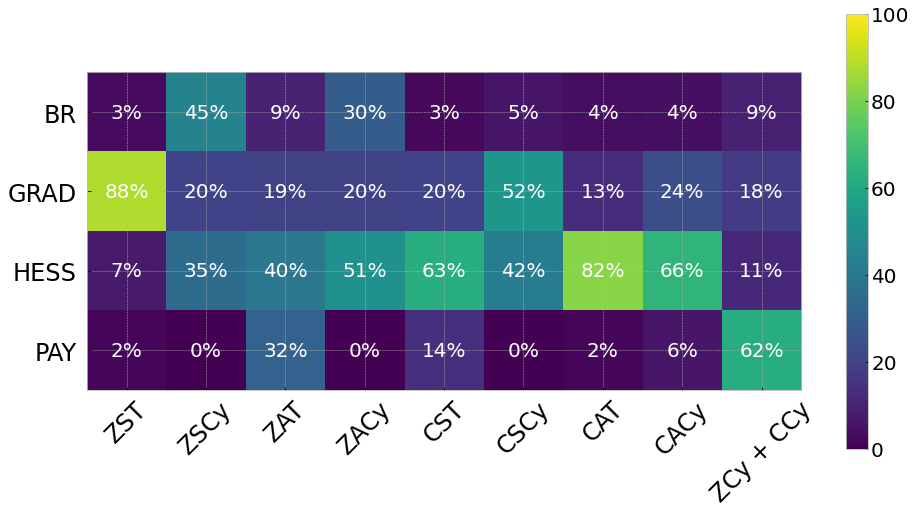}}
  \caption{Percentage contribution towards absolute SHAP importance aggregated across signed values with respect to \emph{gradient learning rates} $G_k=h_k$ for players $k=1,2$ (partial policy outputs) for the four input groups and eight unique game types. These values were computed for the data points at iteration 50 of each trajectory. Columns are normalised to unity with percentages illustrated explicitly in each box as well as by colour.}
  \label{fig:shap:signed_grad}
\end{figure}

\begin{figure}[ht]
  \centering
  \centerline{\includegraphics[scale=0.4]{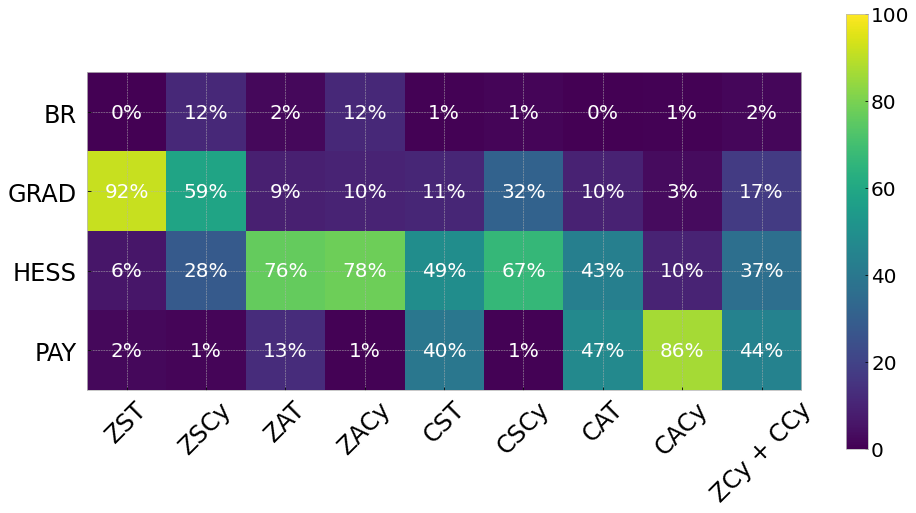}}
  \caption{Percentage contribution towards absolute SHAP importance aggregated across signed values with respect to \emph{Hessian learning rates} $H_k=-h_k \epsilon_k$ (partial policy outputs) for the four input groups and eight unique game types. These values were computed for the data points at iteration 50 of each trajectory. Columns are normalised to unity with percentages illustrated explicitly in each box as well as by colour.}
  \label{fig:shap:signed_hess}
\end{figure}

\begin{figure}[ht]
  \centering
  \centerline{\includegraphics[width=1.\columnwidth]{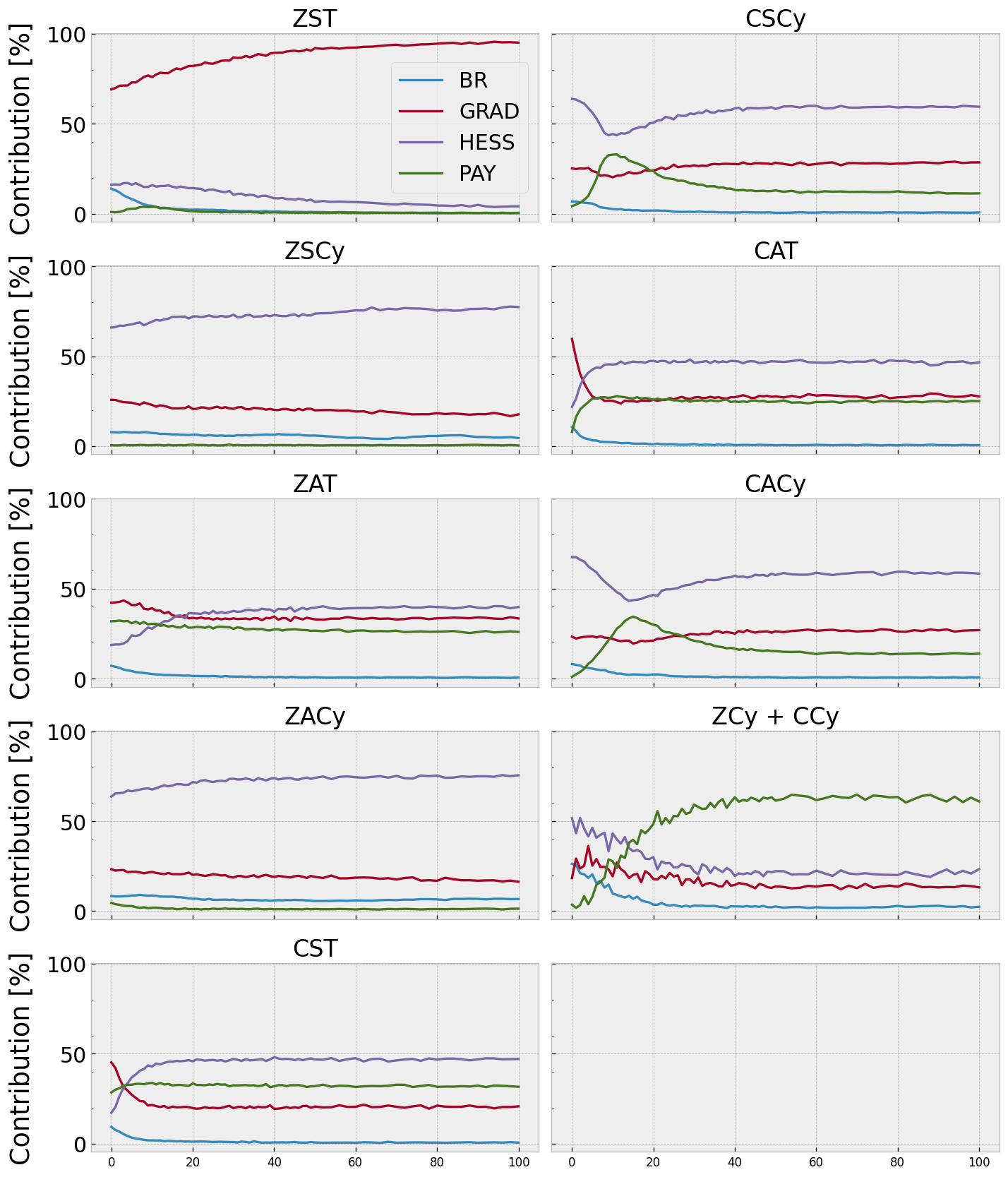}}
  \caption{Percentage contribution towards absolute SHAP importance aggregated across the absolute values with respect to \emph{all learning rates} for the four input groups and eight unique game types over time. In each plot, the curves are normalised to unity at any given iteration.}
  \label{fig:shap:abs_evolution}
\end{figure}

\end{document}